\documentclass[11pt,english]{article}
\usepackage{color}
\usepackage[T1]{fontenc}
\usepackage[latin9]{inputenc}
\usepackage[a4paper]{geometry}
\usepackage{textcomp}
\usepackage{amsthm}
\usepackage{amsmath}
\usepackage{amssymb}
\usepackage{esint}
\usepackage{bbm}
\usepackage{hyperref}
\usepackage{babel}

\usepackage{amscd}

\usepackage{amsfonts}
\usepackage{amsthm}
\usepackage{mathrsfs}
\usepackage{dsfont}

\usepackage{mathtools}

\usepackage{enumerate}
\usepackage{bm}
\usepackage{bbm}
\usepackage{verbatim}

\usepackage{enumitem}
\newlist{abbrv}{itemize}{1}
\setlist[abbrv,1]{label=,labelwidth=1.2in,align=parleft,itemsep=0.1\baselineskip,leftmargin=!}

\makeatletter

\DeclareFontEncoding{LGR}{}{}
\DeclareTextSymbol{\~}{LGR}{126}

\newcommand{\tr}{\text{Tr}}
\newcommand{\vA}{\boldsymbol{\hat{A}}_{\kappa}}
\newcommand{\vAp}{\boldsymbol{\hat{A}}^{+}_{\kappa}}
\newcommand{\vAm}{\boldsymbol{\hat{A}}^{-}_{\kappa}}
\newcommand{\vAc}{\boldsymbol{A}_{\kappa}}
\newcommand{\vApc}{\boldsymbol{A}^{+}_{\kappa}}
\newcommand{\vAmc}{\boldsymbol{A}^{-}_{\kappa}}
\newcommand{\vAd}{\boldsymbol{\mathcal{A}}}
\newcommand{\vAdm}{\boldsymbol{\mathcal{A}}^{-}}
\newcommand{\vAdp}{\boldsymbol{\mathcal{A}}^{+}}
\newcommand{\vE}{\boldsymbol{\hat{E}}_{\kappa}}
\newcommand{\vEp}{\boldsymbol{\hat{E}}^{+}_{\kappa}}
\newcommand{\vEm}{\boldsymbol{\hat{E}}^{-}_{\kappa}}
\newcommand{\vEc}{\boldsymbol{E}_{\kappa}}
\newcommand{\vEpc}{\boldsymbol{E}^{+}_{\kappa}}
\newcommand{\vEmc}{\boldsymbol{E}^{-}_{\kappa}}
\newcommand{\vEdm}{\boldsymbol{\mathcal{E}}^{-}}
\newcommand{\vEdp}{\boldsymbol{\mathcal{E}}^{+}}
\newcommand{\vEd}{\boldsymbol{\mathcal{E}}}

\newcommand{\vG}{\boldsymbol{\hat{G}}}

\newcommand{\vf}{\boldsymbol{f}}
\newcommand{\vgp}{\gamma^{\perp \Lambda}}
\newcommand{\vep}{\boldsymbol{\epsilon}}

\newcommand{\abs}[1]{| #1 |}

\newcommand{\vj}{\boldsymbol{j}}

\newcommand{\scp}[2]{\big\langle #1 , #2 \big\rangle}

\newcommand{\SCP}[2]{\big\langle #1 , #2 \big\rangle}

\newcommand{\bra}[1]{\langle #1 |}

\newcommand{\ket}[1]{| #1 \rangle}

\newcommand{\norm}[1]{\left|\left| #1 \right|\right|}

\renewcommand{\Re}{\mathrm{Re}}
\renewcommand{\Im}{\mathrm{Im}}

\newcommand{\id}{\mathbbm{1}}

\newcommand{\op}{\mathrm{op}}

\newcommand{\be}{\begin{equation}}
\newcommand{\ee}{\end{equation}}

\newtheorem{theorem}{Theorem}[section]
\newtheorem{lemma}[theorem]{Lemma}

\newtheorem{corollary}[theorem]  {Corollary}
\newtheorem{remark}[theorem]  {Remark}
\newtheorem{definition}[theorem] {Definition}

\allowdisplaybreaks[1]

\begin{document}

\title{Derivation of the Maxwell-Schr\"odinger Equations from the Pauli-Fierz Hamiltonian}

\author{
Nikolai Leopold\footnote{
IST Austria (Institute of Science and Technology Austria), Am Campus 1, 3400 Klosterneuburg, Austria. E-mail: {\tt nikolai.leopold@ist.ac.at}} \ and 
Peter Pickl\footnote{Duke Kunshan University, Duke Avenue 8, 215316 Kunshan, China.
\newline
 E-mail: {\tt peter.pickl@dukekunshan.edu.cn}} \footnote{Ludwig-Maximilians-Universit\"at M\"unchen, Theresienstra\ss e 39, {80333} M\"unchen, Germany.
 \newline E-mail: {\tt pickl@math.lmu.de}}
}

\maketitle

\begin{abstract}
\noindent
We consider the spinless Pauli-Fierz Hamiltonian which describes a quantum system of non-relativistic identical particles coupled to the quantized electromagnetic field. We study the time evolution in a mean-field limit where the number $N$ of charged particles gets large while the coupling to the radiation field is rescaled by $1/\sqrt{N}$. At time zero we assume that almost all charged particles are in the same one-body state (a Bose-Einstein condensate) and we assume also the photons to be close to a coherent state. We show that at later times and in the limit $N \rightarrow \infty$ the charged particles as well as the photons exhibit condensation, with the time evolution approximately described by the Maxwell-Schr\"odinger system, which models the coupling of a non-relativistic particle to the classical electromagnetic field.
 Our result is obtained by an extension of the "method of counting", introduced in \cite{pickl1}, to condensates of charged particles in interaction with their radiation field.
\end{abstract}

\noindent
\textbf{MSC class:} 35Q40, 81Q05, 81V10, 82C10   

\section{Setting of the problem}\label{sec:Introduction}
The existence of light quanta, later named photons, was first postulated by Albert Einstein in his renowned paper "On a heuristic point of view about the creation and conversion of light" \cite{einstein}. This led to the invention of Quantum Electrodynamics and supplemented the nature of light, which was formerly described as a wave in classical electromagnetism, with a particle interpretation.
During the last decades the predictions of Quantum Electrodynamics has been tested up to highest accuracy.
Nevertheless, in a lot of situations the corpuscular character of light is subordinate and the second-quantized electromagnetic field can be approximated by a classical field satisfying Maxwell's equations. In this paper, the validity of such an approximation is justified in the mean-field regime. More explicitly, we derive the Maxwell-Schr\"odinger equations from the spinless Pauli-Fierz Hamiltonian. Such a derivation is of great interest to fundamental physics. Moreover, since the applied mean-field approximation reduces the degrees of freedom of the original system tremendously explicit error bounds might also be of interest for numerical simulations.
We consider a system, described by a wave function $\Psi_{N,t} \in \mathcal{H}^{(N)}$, of N identical charged bosons in interaction with a photon field. Here,
\begin{align}
\mathcal{H}^{(N)}  \coloneqq L^2\left( \mathbb{R}^{3N} \right) \otimes \mathcal{F}_p ,
\end{align}
where the photon field is represented by elements of the Fock space 
\begin{align}
 \mathcal{F}_p  \coloneqq  \bigoplus_{n \geq 0} \left[ L^2(\mathbb{R}^3) \otimes \mathbb{C}^2 \right]^{\otimes_s^n} .
\end{align} 
 The subscript $s$ indicates symmetry under interchange of variables. The Hilbert space $\mathfrak{h} \coloneqq L^2(\mathbb{R}^3) \otimes \mathbb{C}^2$ consists of wave functions $f(k,\lambda)$, with wave number $k \in \mathbb{R}^3$ and  helicity $\lambda=1,2$. It is equipped with the inner product 
\begin{align}
\scp{f}{g}_{\mathfrak{h}} \coloneqq \sum_{\lambda=1,2} \int d^3k \, f^*(k,\lambda) g(k,\lambda).
\end{align}
The time evolution of $\Psi_{N,t}$ is governed by the Schr\"odinger equation
 \begin{align}
\label{eq: Pauli Schroedinger equation microscopic}
 i \partial_t \Psi_{N,t} = H_N \Psi_{N,t}, 
 \end{align}
where 
\begin{align}
\label{eq: Pauli-Fierz Hamiltonian}
H_N =& \sum_{j=1}^N  \left( - i  \nabla_j -  \frac{\vA(x_j)}{\sqrt{N}} \right)^2 
+ \frac{1}{N} \sum_{1\leq j < k \leq N} v(x_j - x_k) +
 H_f
\end{align}
denotes the Pauli-Fierz Hamiltonian and
\begin{align}
\vA(x) = \sum_{\lambda=1,2} \int d^3k \,  \tilde{\kappa}(k)
\frac{1}{\sqrt{2 \abs{k}}} \vep_{\lambda}(k) 
\left( e^{ikx} a(k,\lambda) + e^{-ikx} a^*(k,\lambda)  \right)
\end{align}
 the quantized transverse vector potential. The function 
\begin{equation}
\label{eq: Pauli cut off function}
\tilde{\kappa}(k) = (2 \pi)^{- 3/2} \ \id_{\abs{k}\leq \Lambda}(k), \quad
\text{with} \;  \id_{\abs{k}\leq \Lambda}(k) =
\begin{cases} 
1 &\text{if } \abs{k} \leq \Lambda , \\
0 &\text{otherwise}, 
\end{cases}
\end{equation}
cuts off the high frequency modes of the radiation field.
There are two real  polarization vectors $\vep_1(k)$ and $\vep_2(k)$  with
\begin{align}
\label{eq: polarization vectors}
\abs{\vep_1(k)} =  \abs{\vep_2(k)} = 1 , \;
\vep_1(k) \cdot k = \vep_2(k) \cdot k 
= \vep_1(k) \cdot \vep_2(k) = 0.
\end{align}
 The operator valued distributions $a(k,\lambda)$ and $a^*(k,\lambda)$ $(k \in \mathbb{R}^3, \lambda \in \{1,2\})$ are the usual pointwise annihilation and creation operators in $\mathcal{F}_p$, satisfying
\begin{align}
\label{eq: canonical commutation relation}
[a(k,\lambda), a^*(l,\mu) ] &= \delta_{\lambda,\mu} \delta(k-l), \quad
[a(k,\lambda), a(l,\mu) ] = 
[a^*(k,\lambda), a^*(l,\mu) ] = 0.
\end{align}
The energy of the photon field is given by
\begin{align}
H_f &= \sum_{\lambda=1,2} \int d^3k \, \abs{k} a^*(k,\lambda) a(k,\lambda)
\end{align}
and the potential $v$ describes a direct interaction between the charged particles.\\

\newpage

\noindent
We assume:
\begin{itemize}
\item[(A1)] 
The (repulsive) interaction potential $v$ is a positive, real, and even function satisfying 
\begin{align}
\norm{v}_{L^2+L^{\infty}} = 
\inf_{v=v_1 + v_2} \{ \norm{v_1}_{L^2(\mathbb{R}^3)}
+ \norm{v_2}_{L^{\infty}(\mathbb{R}^3)}  \} < \infty
\end{align}
such that the Pauli-Fierz Hamiltonian $H_N$ is self-adjoint on the domain 
$ \mathcal{D}(H_N) \coloneqq \mathcal{D}(\sum_{i=1}^N - \Delta_i + H_f ) $  (see \cite{hiroshima} and \cite[p.164]{spohn}).
\end{itemize}
The mean-field scaling 1/N in front of the interaction potential and the scaling $1/\sqrt{N}$ in front of the vector potential ensure that the kinetic and potential energy of  $H_N$ are of the same order. 
For simplicity, we are first interested in the evolution of initial states of the product form
\begin{align}
\label{eq: Pauli initial product state}
\varphi_0^{\otimes N} \otimes W(\sqrt{N} \alpha_0) \Omega.
\end{align}
Here, $\Omega$ denotes the vacuum in $\mathcal{F}_p$ and $W(f)$ (with $f \in \mathfrak{h}$) is the unitary Weyl operator
\begin{align}
\label{eq: Pauli Weyl operator}
W(f) \coloneqq \exp \Big( \sum_{\lambda=1,2} \int d^3k \, f(k,\lambda) a^*(k,\lambda) - f^*(k,\lambda) a(k,\lambda) \Big).
\end{align} 
This choice of initial data corresponds to situations in which both the charged particles and the photons exhibit condensation.
 Due to different types of interactions, correlations take place and the time evolved state will no longer have an exact product structure. However, for large $N$ and times of order one it can be approximated, in a sense specified below, by a state of the product form $\varphi_t^{\otimes N} \otimes W(\sqrt{N} \alpha_t) \Omega$, where 
\begin{align}
\label{eq: Pauli classical mode function}
\abs{k}^{1/2} \alpha_t(k,\lambda) \coloneqq \frac{1}{\sqrt{2}}  \vep_{\lambda}(k) \cdot \left( \abs{k} \mathcal{FT}[ \boldsymbol{A}](k,t) - i \mathcal{FT}[\boldsymbol{E}](k,t) \right)
\end{align}
and $(\varphi_t,\boldsymbol{A}(t),\boldsymbol{E}(t))$
solve the Maxwell-Schr\"odinger system\footnote{ Hereby, $(\kappa * \boldsymbol{A})(x,t) = \int d^3k \, e^{ikx} \tilde{\kappa}(k) \boldsymbol{A}(k,t)$ and $\tilde{f}$ as well as $\mathcal{FT}[f]$ are used to denote the Fourier transform of a function $f$.}
\begin{align}
\label{eq: Pauli Hartree-Maxwell system 2}
\begin{cases}
i \partial_t \varphi_t(x) &= \left( \left(-i \nabla - (\kappa * \boldsymbol{A})(x,t) \right)^2 + (v * \abs{\varphi_t}^2)(x)  \right) \varphi_t(x) ,    \\
\nabla \cdot \boldsymbol{A}(x,t) &= 0 ,  \\ 
\partial_t \boldsymbol{A}(x,t) &= -  \boldsymbol{E}(x,t),   \\
\partial_t \boldsymbol{E}(x,t) &= \left( - \Delta \boldsymbol{A}\right) (x,t) -    \left( 1 - \nabla \text{div} \Delta^{-1} \right) \left( \kappa * \vj_t \right)(x) , \\
\vj_t(x) &= 2 \left(  \Im(\varphi_t^* \nabla \varphi_t)(x) - \abs{\varphi_t}^2(x)  (\kappa * \boldsymbol{A})(x,t) \right)
\end{cases}
\end{align}
with initial datum
\begin{align}
\label{eq: Pauli initial data Hartree Maxwell}
\begin{cases}
&\varphi_0 , \\
&\boldsymbol{A}(x,0) =  (2 \pi)^{-3/2}   \sum_{\lambda=1,2} \int d^3k \, \frac{1}{\sqrt{2 \abs{k}}} \vep_{\lambda}(k) 
\left( e^{ikx} \alpha_0(k,\lambda) + e^{-ikx} \alpha_0^{*}(k,\lambda)  \right)  ,  \\
&\boldsymbol{E}(x,0) =  (2 \pi)^{-3/2}   \sum_{\lambda=1,2} \int d^3k  \,  \sqrt{\frac{\abs{k}}{2}} \vep_{\lambda}(k)  i
\left( e^{ikx} \alpha_0(k,\lambda) - e^{-ikx} \alpha_0^{*}(k,\lambda)   \right)  .
 \end{cases}
\end{align}
These equations determine the time evolution of  a single quantum particle interacting with the classical electromagnetic field it generates. The solution theory of this system is well studied, see \cite{nakamurawada} and references therein.

\section{Main result}
The physical situation we are interested in is the dynamical description of a Bose-Einstein condensate of charged particles. We start with an initial wave function of product form \eqref{eq: Pauli initial product state} (a condition that will be relaxed later) and show that the condensate is stable over time, i.e. correlations are small at later times.
Let $\Psi_{N,t} \in \left(L^2_s \left( \mathbb{R}^{3N} \right) \otimes \mathcal{F}_p \right) \cap  \mathcal{H}^{(N)}$ with  $\norm{\Psi_{N,t}}=1$. On the Hilbert space $L^2(\mathbb{R}^3)$, define the "one-particle reduced density matrix of the charged particles"  by
\begin{align}
\label{eq: definition reduced one-particle matrix charged particles}
\gamma_{N,t}^{(1,0)} \coloneqq \tr_{2,\ldots, N} \otimes \tr_{\mathcal{F}} \ket{\Psi_{N,t}} \bra{\Psi_{N,t}},
\end{align}
where $\tr_{2,\ldots, N}$  denotes the partial trace over the coordinates $x_2,\ldots, x_N$ and $\tr_{\mathcal{F}}$ the trace over Fock space. The charged particles of the many-body state $\Psi_{N,t}$ are said to exhibit complete asymptotic Bose-Einstein condensation at time $t$, if there exists 
$\varphi_t \in L^2(\mathbb{R}^3)$ with $\norm{\varphi_t}=1$, such that
\begin{align}
\label{eq: convergence reduced one-particle matrix charged particles}
\tr_{L^2(\mathbb{R}^3)} \abs{\gamma_{N,t}^{(1,0)} - \ket{\varphi_t} \bra{\varphi_t}}  \rightarrow 0,
\end{align}
as $N \rightarrow \infty$. Such $\varphi_t$ is called the condensate wave function. For other indicators of condensation and their relation we refer to \cite{michelangeli2}. Given  $\Psi_{N,t} \in \mathcal{D}(H_f)$ with  $\norm{\Psi_{N,t}}=1$, we introduce the "one-particle reduced energy matrix of the photons"  with kernel
\begin{align}
\label{eq: definition reduced one-particle matrix photon}
\gamma_{N,t}^{(0,1)}(k,\lambda;k',\lambda') \coloneqq N^{-1} \abs{k}^{1/2} \abs{k'}^{1/2} \scp{\Psi_{N,t}}{ a^*(k',\lambda')  a(k,\lambda) \Psi_{N,t}}_{\mathcal{H}^{(N)}}  .
\end{align}
$\gamma_{N,t}^{(0,1)}$ is a positive trace class operator on $\mathfrak{h}$ with $\tr_{\mathfrak{h}}(\gamma_{N,t}^{(0,1)})= N^{-1} \scp{\Psi_{N,t}}{H_f \Psi_{N,t}}_{\mathcal{H}^{(N)}}$.
It is important to note, that \eqref{eq: definition reduced one-particle matrix photon} differs from the usual definition (e.g. \cite[p.8]{rodnianskischlein}) by the weight factor  $\abs{k}^{1/2} \abs{k'}^{1/2} \scp{\Psi_{N,t}}{\mathcal{N} \Psi_{N,t}}_{\mathcal{H}^{(N)}}/N$ with $\mathcal{N}$ being the number of photons operator. Our choice ensures that we neglect photons with small energies and measure only deviations from the photon field that are at least of order $N$. This is reasonable because due to the scaled coupling many photon states with a mean particle number smaller than of order $N$ only have a subleading effect on the dynamics of the charged particles. We say the photons exhibit asymptotic Bose-Einstein condensation, if there exists a state $u_t \in \mathfrak{h}$, such that
\begin{align}
\label{eq: convergence reduced one-particle matrix photon}
\tr_{\mathfrak{h}} \abs{\gamma_{N,t}^{(0,1)} - \ket{u_t} \bra{u_t}} \rightarrow 0,
\end{align}
 as $N \rightarrow \infty$. \\
In order to prove our main theorem, we have to require regularity assumptions on the solutions of the Maxwell-Schr\"odinger system.
\begin{definition}
\label{def: Pauli assumptions on the solutions of the Maxwell-Schroedinger system}
Let $m \in \mathbb{N}$, and $H^m(\mathbb{R}^3, \mathbb{C}^k)$ denote the Sobolev space of order $m$ with norm
$\norm{f}_{H^m(\mathbb{R}^3,\mathbb{C}^k)}^2 = \sum_{i=1}^k \int d^3k \, (1 + \abs{k}^2)^m \abs{\mathcal{FT}[f^i](k)}^2 $. We define the following set of solutions of the Maxwell-Schr\"odinger equations:
\begin{align}
\label{eq: Pauli assumptions on the solutions of the Maxwell-Schroedinger system}
(\varphi_t,\boldsymbol{A}(t),\boldsymbol{E}(t)) \in \mathcal{G} \Leftrightarrow  &(a) \quad (\varphi_0,\boldsymbol{A}(0),\boldsymbol{E}(0)) \; \text{is given by} \; \eqref{eq: Pauli initial data Hartree Maxwell} \; \text{with} \; \alpha_0 \in \mathfrak{h} \; \text{and} \;  \varphi_0 \in L^2(\mathbb{R}^3)
\nonumber \\
&(b) \quad (\varphi_t,\boldsymbol{A}(t),\boldsymbol{E}(t))  \; \text{is a $L^2 \oplus L^2 \oplus L^2$ solution of \eqref{eq: Pauli Hartree-Maxwell system 2}} 
\nonumber \\
  &(c) \quad (\varphi_t,\boldsymbol{A}(t),\boldsymbol{E}(t)) \in H^2(\mathbb{R}^3,\mathbb{C}) \oplus H^2(\mathbb{R}^3,\mathbb{C}^3) \oplus H^1(\mathbb{R}^3,\mathbb{C}^3)
  \nonumber \\
  &(d) \quad \norm{\varphi_t}_{L^2(\mathbb{R}^3)} =1
  \nonumber \\
 & \qquad  \; \,  \text{for all $t \geq 0$. }
\end{align}
\end{definition}
\noindent
We expect these assumptions to follow from appropriately chosen initial data. In the absence of a cutoff function and $v$ being the Coulomb potential it has for example been shown in \cite{nakamurawada} that the Maxwell-Schr\"odinger system is globally well-posed in  the space\footnote{The direct sum of the Sobolev spaces refers to $(\varphi_t, \boldsymbol{A}(t), \boldsymbol{E}(t))$.} $C(\mathbb{R}_t, H^2(\mathbb{R}^3,\mathbb{C}) \oplus H^2(\mathbb{R}^3,\mathbb{C}^3) \oplus H^1(\mathbb{R}^3,\mathbb{C}^3))$.
\\
\noindent 
Moreover, we have to introduce the energy mode function of the electromagnetic field\footnote{
The name energy mode function is motivated from the fact that $\norm{u_t}_{\mathfrak{h}}^2 = \scp{u_t}{u_t}_{\mathfrak{h}}$ is the energy of the electromagnetic field.
It should be noted that $u_t \in \mathfrak{h}$ and $\mathcal{E}_M[\varphi_t,u_t] < \infty$ follow from $(\varphi_t,\boldsymbol{A}(t),\boldsymbol{E}(t)) \in \mathcal{G}$ (see \eqref{eq: Pauli regularity 1} and \eqref{eq: Pauli regularity 2}).} 
\begin{align}
\label{eq: Pauli energy mode function}
u_t(k,\lambda) \coloneqq  \abs{k}^{1/2} \alpha(k, \lambda)  \coloneqq \frac{1}{\sqrt{2}} \vep_{\lambda}(k) \cdot \left( \abs{k} \mathcal{FT}[\boldsymbol{A}](k,t) - i  \mathcal{FT}[\boldsymbol{E}](k,t) \right).
\end{align}
and the energy functional of the Maxwell-Schr\"odinger system
\begin{align}
\label{eq: Pauli energy functional of the HM system}
\mathcal{E}_M\left[\varphi_t, u_t \right]
&\coloneqq \norm{\left(- i \nabla - \vAc(t) \right) \varphi_t}^2 + 1/2 \scp{\varphi_t}{\left( v * \abs{\varphi_t}^2 \right) \varphi_t}  + \norm{u_t}_{\mathfrak{h}}^2. 
\end{align}

\begin{theorem}
\label{theorem: Pauli main theorem}
Let $v$ satisfy (A1),  $\varphi_0 \in L^2(\mathbb{R}^3)$ with $\norm{\varphi_0} =1$,   $u_0 \in \mathfrak{h}$ so that $\alpha_0 = \abs{k}^{-1/2}  u_0 \in \mathfrak{h}$ 
and $\Psi_{N,0} \in \mathcal{D}(H_N)\cap \left(L_s^2(\mathbb{R}^{3N})\otimes \mathcal{F}_p \right)$ such that
\begin{align}
a_N &\coloneqq \text{Tr}_{L^2(\mathbb{R}^3)} \abs{\gamma_{N,0}^{(1,0)} - \ket{\varphi_0} \bra{\varphi_0}} \rightarrow 0, \\
b_N &\coloneqq N^{-1} \scp{W^{-1}(\sqrt{N} \alpha_0)\Psi_{N,0}}{H_f W^{-1}(\sqrt{N} \alpha_0)\Psi_{N,0}}_{\mathcal{H}^{(N)}}
\rightarrow 0 \; \text{and}  \\
\label{eq: Pauli initial condition in main theorem variance of the energy}
c_N  &\coloneqq \norm{\left(N^{-1} H_N - \mathcal{E}_M\left[\varphi_0, u_0  \right] \right) \Psi_{N,0} }_{\mathcal{H}^{(N)}}^2 \rightarrow 0
\end{align}
as $N \rightarrow \infty$.
Let $\Psi_{N,t}$ be the unique solution of \eqref{eq: Pauli Schroedinger equation microscopic}, $(\varphi_t,\boldsymbol{A}(t),\boldsymbol{E}(t)) \in \mathcal{G}$ and $u_t$ be defined by \eqref{eq: Pauli energy mode function}.
Then, there exists a monotone increasing function $C(s)$ of the norms  $\norm{\varphi_{s}}_{H^2(\mathbb{R}^2, \mathbb{C})}$, $\norm{\boldsymbol{A}(s)}_{H^2(\mathbb{R}^3, \mathbb{C}^3)}$ and $\norm{\boldsymbol{E}(s)}_{L^2(\mathbb{R}^3, \mathbb{C}^3)}$ such that
\begin{align}
 \label{eq: main theorem 1}
\text{Tr}_{L^2(\mathbb{R}^3)} \abs{\gamma_{N,t}^{(1,0)} - \ket{\varphi_t} \bra{\varphi_t}} &\leq
\sqrt{a_N + b_N + c_N + N ^{-1} } \, \Lambda e^{\Lambda^4 \int_0^t ds \, C(s)} , \\
\label{eq: main theorem 2}
\text{Tr}_{\mathfrak{h}} \abs{\gamma_{N,t}^{(0,1)} - \ket{u_t}\bra{u_t}}  &\leq  
\sqrt{a_N + b_N + c_N + N^{-1} } \,  \Lambda C(s) e^{\Lambda^4 \int_0^t ds \, C(s)} .
\end{align}
for any $t \geq 0$.
In particular, for $\Psi_{N,0} = \varphi_{0}^{\otimes N} \otimes  W(\sqrt{N} \alpha_0) \Omega$ one obtains
\begin{align}
 \label{eq: main theorem 3}
\text{Tr}_{L^2(\mathbb{R}^3)} \abs{\gamma_{N,t}^{(1,0)} - \ket{\varphi_t} \bra{\varphi_t}} &\leq
 N^{-1/2}  \Lambda^2 e^{\Lambda^4 \int_0^t ds \, C(s)}, \\
\label{eq: main theorem 4}
\text{Tr}_{\mathfrak{h}} \abs{\gamma_{N,t}^{(0,1)} - \ket{u_t}\bra{u_t}}  &\leq
 N^{-1/2}  \Lambda^2 C(s) e^{\Lambda^4 \int_0^t ds \, C(s)}.
\end{align}
\end{theorem}

\begin{remark}
Assumption (A1) allows to consider the Coulomb potential $v(x) =  \abs{x}^{-1}$. The requirements on the interaction potential can easily be relaxed  because our estimates only rely on the finiteness of $\norm{v * \abs{\varphi_t}^2}_{L^{\infty}(\mathbb{R}^3)}$ and  $\norm{v^2 * \abs{\varphi_t}^2}_{L^{\infty}(\mathbb{R}^3)}$. This is captured by (A1) and $\varphi_t \in H^2(\mathbb{R}^3)$ but also by other means.
\end{remark}

\begin{remark}
For simplicity we apply the mean-field scaling $1/N$ in front of the direct interaction. Using techniques from \cite{pickl2} and \cite{picklgp3d} it seems possible to treat the direct interaction also in the NLS or Gross-Pitaevskii regime.
\end{remark}

\begin{remark}
The minimal coupling term in the Pauli-Fierz Hamiltonian leads to an interaction between the charges and the radiation field which is more singular than for example in the Nelson model. This makes it difficult to control the number of photons with small energies during the time evolution and causes technical problems in the estimates.
To overcome these difficulties we we neglect contributions from photons with small energies and restrict our initial data to a subspace of many-body states whose energy per particle only fluctuates little around the energy of the effective system. This is the reason why we consider the one-particle reduced energy matrix instead of the one-particle reduced density matrix of the photons. Moreover, it explains the appearance of condition \eqref{eq: Pauli initial condition in main theorem variance of the energy}.
\end{remark}

\begin{remark}
The ultraviolet cutoff is essential in our derivation but can be chosen N-dependent.
\end{remark}

\section{Comparison with the literature}  
  
Derivations of classical field equations from Many-body Quantum Dynamics has been established in a series of works:
In \cite{ginibrenironivelo}, Ginibre, Nironi and Velo derived the Schr\"odinger-Klein-Gordon system of equations from the Nelson model with cutoff. They considered a mean-field limit where a finite number of charged particles interacts with a coherent state of gauge bosons whose particle number goes to infinity. Falconi \cite{falconi}  derived  the Schr\"odinger-Klein-Gordon system of equations in a mean-field limit where both the number of the charged particles and the gauge bosons go to infinity. 
Making use of a Wigner measure approach Ammari and Falconi \cite{ammarifalconi} were able to establish the classical limit of the renormalized Nelson model without cutoff.
The replacement of quantized radiation fields by classical interactions has also been justified in other limits.
Teufel \cite{teufel} considered the adiabatic limit of the Nelson model and showed that the interaction mediated by the quantized radiation field is well approximated by a direct Coulomb interaction. In \cite{frankschlein} and \cite{frankgang}, Frank, Gang and Schlein showed that in the strong coupling limit the dynamics of a polaron is described by an effective equation, in which the phonon field is treated as a classical field.
Knowles \cite{knowles2} analyzed a finite number of heavy particles in a strong radiation field and derived the Newton-Maxwell equations from the Pauli Fierz Hamiltonian.
In \cite{sindelka} it is shown  that the semiclassical set of coupled Maxwell-Schr\"odinger equations is obtained by neglecting certain terms of the Pauli-Fierz Hamiltonian. 
To our best knowledge, this is the first rigorous result concerning a mean-field limit of the Pauli-Fierz Hamiltonian. This work continues the master thesis \cite{vytas}.

\section{Notations}
We set Planck's constant $\hbar$, the speed of light $c$, the charge $e$, and twice the mass of the particles $2m$ equal to one.
Except in definitions, results, and where confusion might be possible, we refrain from indicating the explicit dependence of a quantity on the time $t$. We use the notations $\varphi(t)$ and $\varphi_t$ interchangeably to denote a quantity $\varphi$ at time t. The symbol $C$ is used as a generic positive constant independent of $t$, $N$ and $\Lambda$. We use expressions like $C(\norm{\varphi}_{H^2(\mathbb{R}^3)}, \norm{\boldsymbol{A}}_{L^2(\mathbb{R}^3, \mathbb{C}^3)})$ to denote positive monotone increasing function of the norms indicated.
Both $\tilde{f}$ and $\mathcal{FT}[f]$ stand for the Fourier transform of $f$. 
With a slight abuse of notation $\boldsymbol{A}$ and $\boldsymbol{E}$  denote the vector potential and the electric field, but also their respective Fourier transforms. If we write $\boldsymbol{A}(t)$ or $\boldsymbol{E}(t)$,  we always refer to the coordinate representation of the electromagnetic fields.
Furthermore, we apply the shorthand notation $\vAc(x,t) = \left( \kappa * \boldsymbol{A} \right)(x,t)$. 
$H^m(\mathbb{R}^3, \mathbb{C}^k)$ stand for the Sobolev space of order m with norm
$\norm{f}_{H^m(\mathbb{R}^3,\mathbb{C}^k)}^2 = \sum_{i=1}^k \int d^3k \, (1 + \abs{k}^2)^m \abs{\mathcal{FT}[f^i](k)}^2 $.
To simplify the notation we use $H^m(\mathbb{R}^3)$ for $H^m(\mathbb{R}^3, \mathbb{C}^k)$ with $k \in \{1,2,3 \}$.
 $\norm{A}_{HS} = \sqrt{Tr A^* A}$ stands for the Hilbert-Schmidt norm and $\SCP{\cdot}{\cdot}$ denotes the scalar products on $\mathcal{H}^{(N)}$, $L^2(\mathbb{R}^3)$ and $\mathfrak{h}$. Furthermore, we use the shorthand notation $\SCP{\cdot}{\cdot}_{;y} = \int d^3y \, \SCP{\cdot}{\cdot}$ and $\norm{\cdot}_{;y}^2 = \int d^3y \, \SCP{\cdot}{\cdot}$. In order to stress the connection between the annihilation operator and the mode function of the electromagnetic field we mostly write $\abs{k}^{1/2} \alpha_t(k,\lambda)$ instead of $u_t(k,\lambda)$.

\section{Organization of the proof}
\label{section: Pauli organization of the proof}

Our result is obtained by an extension of the "method of counting", introduced in \cite{pickl1}, to condensates of charged particles in interaction with their radiation field \cite{leopold}. The key idea is not to prove condensation in terms of reduced density matrices but to consider a different indicator of condensation. More specific, we introduce a  functional $\beta(t): \mathcal{D}(H_N) \times H^2(\mathbb{R}^3) \times \mathfrak{h}  \rightarrow \mathbb{R}_0^+ $ with the properties:
\begin{itemize}
\item[(a)] $\beta(0) \coloneqq \beta\left[\Psi_0, \varphi_0, u_0\right] \rightarrow 0$ as $N \rightarrow 0$ for appropriately chosen initial data.
\item[(b)]  For each $t \in \mathbb{R}_0^+$ we are able to control the time-derivative of the functional by $\abs{d_t \beta(t)} \leq C (\beta(t) + N^{-1})$. Then, $\beta(t) \leq e^{Ct} (\beta(t) + N^{-1})$ follows by Gronwall's Lemma. 
\item[(c)] $\beta(t) \rightarrow 0$ as $N \rightarrow \infty$ implies \eqref{eq: convergence reduced one-particle matrix charged particles} and \eqref{eq: convergence reduced one-particle matrix photon}.
\end{itemize}

\noindent
The proof is organized as follows:
\begin{itemize}

\item[(a)] In Section~\ref{section: Pauli the functional} we define the counting functional. Afterwards, we show that  convergence of the functional to zero in the limit $N \rightarrow \infty$ implies condensation in terms of reduced density matrices.

\item[(b)]  In section~\ref{section: Pauli estimates on the time derivative} we control the growth of $\beta$ by means of a Gronwall estimate. To this end, we provide preliminary estimates and control the time derivative of $\beta$.

\item[(c)]  Then, we relate the value of the functional at time zero to the initial data of Theorem~\ref{theorem: Pauli main theorem}. 

\end{itemize}

\noindent
Subsequently, we require $(\varphi_t,\boldsymbol{A}(t),\boldsymbol{E}(t)) \in \mathcal{G}$. This implies
\begin{align}
\label{eq: Pauli regularity 1}
\quad  \norm{\nabla \varphi_t} <& \infty ,
\quad  \norm{\Delta \varphi_t} < \infty ,
\quad  \norm{\varphi_t}_{\infty} < \infty ,
\quad \norm{\vAc(t)}_{\infty} < \infty ,  \\
\label{eq: Pauli regularity 2}
\norm{u_t}_{\mathfrak{h}} <& \infty  , 
\quad ||\abs{\cdot}^{1/2} u_t ||_{\mathfrak{h}} =  \big( \sum_{\lambda=1,2} \int \, d^3k \abs{k}^2  \abs{\alpha_t(k,\lambda)}^2 \big)^{1/2} < \infty.
\end{align}
\begin{proof}
Assuming  that $(\varphi_t,\boldsymbol{A}(t),\boldsymbol{E}(t)) \in \mathcal{G}$ we have $\varphi_t \in H^2(\mathbb{R}^3)$ and the first three relations follow from Sobolev inequalities. In order to show the finiteness of $\norm{\vAc(t)}_{\infty}$ we define the functions
\begin{align}
\tilde{\kappa}_{<}(k) \coloneqq (2 \pi)^{-3/2} \id_{{\abs{k} \leq 1}}(k), 
\qquad
\tilde{\kappa}_{>}(k) \coloneqq  (2 \pi)^{-3/2} \abs{k}^{-2} \id_{{1 \leq \abs{k} \leq \Lambda}}(k)
\end{align}
and continue with 
\begin{align}
\vAc(x,t) &= (2 \pi)^{-3/2} \int d^3k \, e^{ikx} \id_{\abs{k} \leq \Lambda}(k) \boldsymbol{A}(k,t)  
= (2 \pi)^{-3/2} \int d^3k \, e^{ikx} \id_{\abs{k} \leq 1}(k) \boldsymbol{A}(k,t) 
\nonumber  \\
&+ (2 \pi)^{-3/2} \int d^3k \, e^{ikx} \abs{k}^{-2} \id_{1 \leq \abs{k} \leq \Lambda}(k)  \abs{k}^2 \boldsymbol{A}(k,t) 
\nonumber \\
&= \left( \kappa_{<} * \boldsymbol{A} \right)(x,t)
- \left( \kappa_{>}  *  \Delta \boldsymbol{A} \right)(x,t).
\end{align}
So if we use Young's inequality, $\norm{\kappa_{<}}_2^2 = (6 \pi^2)^{-1}$ and $\norm{\kappa_{>}}_2^2 \leq (4 \pi^2)^{-1}$, we get 
\begin{align}
\label{eq: Pauli infitinity norm of the classical vector potential with cutoff}
\norm{\vAc(t)}_{\infty} 
&\leq \norm{\left( \kappa_{<} * \boldsymbol{A} \right)(t)}_{\infty}
+  \norm{\left( \kappa_{>} * \Delta \boldsymbol{A} \right)(t)}_{\infty} 
\nonumber \\
&\leq \norm{\kappa_{<}} \norm{\boldsymbol{A}(t)}
+ \norm{\kappa_{>}} \norm{\Delta \boldsymbol{A}(t)}
< \norm{\boldsymbol{A}(t)}_{H^2(\mathbb{R}^3)}.
\end{align}
By means of 
\begin{align}
\sum_{\lambda=1,2} \vep_{\lambda}^i(k) \vep_{\lambda}^j(k) = \delta_{ij} - \frac{k_i k_j}{\abs{k}^2}
\end{align}
one easily shows
\begin{align}
\label{eq: Pauli dependence of the field energy on A and E}
\norm{u_t}_{\mathfrak{h}}^2
&= 1/2 \int d^3k \, \left( \abs{k}^2  \boldsymbol{A}^2(k,t) + \boldsymbol{E}^2(k,t)  \right) 
\leq  \norm{\boldsymbol{A}(t)}_{H^1(\mathbb{R}^3)}^2 +  \norm{\boldsymbol{E}(t)}^2 ,
\nonumber \\
||\abs{\cdot}^{1/2} u_t ||_{\mathfrak{h}}^2
&= 1/2 \int d^3k \,  \left( \abs{k}^3 \boldsymbol{A}^2(k,t) + \abs{k} \boldsymbol{E}^2(k,t)  \right) 
\nonumber \\
&\leq \norm{\boldsymbol{A}(t)}_{H^1(\mathbb{R}^3)} \norm{ \boldsymbol{A}(t)}_{H^2(\mathbb{R}^3)} +  \norm{\boldsymbol{E}(t)} \norm{ \boldsymbol{E}(t)}_{H^1(\mathbb{R}^3)} .
\end{align}
\end{proof}

\section{The counting functional}
\label{section: Pauli the functional}
In this section, we introduce a new indicator of condensation referred to as the "counting functional". Our system under consideration \eqref{eq: Pauli Schroedinger equation microscopic} describes the interaction between charged particles and the quantized electromagnetic field.
Initially, we assume the charges and photons to exhibit condensation and we would like to show that both condensates are stable over time. 
In case of the charges, this is done by means of a functional, denoted by $\beta^a$, which counts for each time $t$ the relative number of charges which are not in the state $\varphi_t$. 
Under suitable conditions on the photon field it is then possible to show that the rate of particles which leave the condensate is small, if initially almost all particles are in the state $\varphi_0$. The situation is different for the radiation field because the number of photons is not a conserved quantity. On that account not only existing photons gets correlated but also new photons are created or destroyed. One should note that the high frequency modes of the radiation field do not interact with the charges due to the ultraviolet cutoff \eqref{eq: Pauli cut off function} and evolve according to the free evolution. This is why neither the number of photons changes nor the photon state shows correlations for wave-numbers $\abs{k} \geq \Lambda$. However in the long wave-length sector of $\mathcal{F}_b$ correlations take place and the number of photons varies. To show that the photon field remains coherent we introduce the functional $\beta^b$ measuring  for each time $t$ the fluctuations of the photon field around the classical mode function.
An additional factor of $\abs{k}$ in the integral implies that we neglect contributions from photons with small energies. The main difficulties in our derivation arise from the minimal coupling term in the Pauli-Fierz Hamiltonian. On that account we have to control expectation values of certain unbounded operators, see Subsection~\ref{subsection: Pauli preliminary estimates}.
 This is established by $\beta^c$ which restricts our consideration to a subspace of many-body states whose energy per particle only fluctuates little around the energy functional of the effective system. \\
In order to define the counting functional we introduce the  projectors $p^{\varphi_t}_j$ and $q^{\varphi_t}_j$.

\begin{definition}
For any $N \in \mathbb{N}$, $\varphi_t \in L^2(\mathbb{R}^3)$ with $\norm{\varphi_t}=1$ and $1 \leq j \leq N$ we define the time-dependent projectors
$p_j^{\varphi_t}: L^2(\mathbb{R}^{3N}) \rightarrow  L^2(\mathbb{R}^{3N})$ and
$q_j^{\varphi_t}: L^2(\mathbb{R}^{3N}) \rightarrow  L^2(\mathbb{R}^{3N})$  by
\begin{align}
p_j^{\varphi_t} f(x_1, \ldots, x_N)
\coloneqq \varphi_t(x_j) \int d^3x_j \, \varphi_t^*(x_j) f(x_1, \ldots, x_N) 
\quad \text{for all} \; f \in L^2(\mathbb{R}^{3N})
\end{align}
and $q_j^{\varphi_t} \coloneqq 1 - p_j^{\varphi_t}$.\footnote{For ease of notation we mostly omit the superscript $\varphi_t$ in the following. Additionally, we use the bra-ket notation $p_j^{\varphi_t} = \ket{\varphi_t(x_j)} \bra{\varphi_t(x_j)}$} 
\end{definition}

\noindent
The counting functional is defined by 
\begin{definition}
\label{definition: beta-functional}
Let $\Psi_{N,t} \in \mathcal{D}(H_N)$, $\varphi_t \in H^2(\mathbb{R}^3)$, $u_t \in \mathfrak{h}$ and $\mathcal{E}_M\left[\varphi_t, u_t \right]$ defined by \eqref{eq: Pauli energy functional of the HM system}. Then
\begin{align}
\beta^a(\Psi_{N,t},\varphi_t) &\coloneqq  \SCP{\Psi_{N,t}}{q_1 \otimes \id_{\mathcal{F}_p} \, \Psi_{N,t}} , 
\nonumber \\
\beta^b(\Psi_{N,t},u_t) &\coloneqq \sum_{\lambda=1,2} \int d^3k \,  \norm{\left( \abs{k}^{1/2} \frac{ a(k,\lambda)}{\sqrt{N}} - u_t(k,\lambda) \right) \Psi_{N,t}}^2
\nonumber \\
&= \sum_{\lambda=1,2} \int d^3k \, \abs{k} \SCP{\left( \frac{a(k,\lambda)}{\sqrt{N}} - \alpha_t(k,\lambda) \right) \Psi_{N,t}}{ \left( \frac{a(k,\lambda)}{\sqrt{N}} - \alpha_t(k,\lambda) \right) \Psi_{N,t}}, 
\nonumber  \\
\beta^c(\Psi_{N,t},\varphi_t, u_t) &\coloneqq \SCP{\left( \frac{H_N}{N} - \mathcal{E}_M[\varphi_t,u_t] \right) \Psi_{N,t}}{\left( \frac{H_N}{N} - \mathcal{E}_M[\varphi_t, u_t] \right) \Psi_{N,t}}.
\end{align}
The functional  $\beta: \mathcal{D}(H_N) \times H^2(\mathbb{R}^3) \times \mathfrak{h}  \rightarrow \mathbb{R}_0^+ $
is given by
$ \beta \coloneqq \beta^a + \beta^b + \beta^c$.
\end{definition}

\noindent
The functional $\beta^a$ was already used in \cite{anapolitanosmott, pickl4, knowles, michelangeli, pickl1, pickl2, picklgp3d} and others to derive the Hartree and Gross-Pitaevskii equation, while $\beta^b$ and $\beta^c$ are introduced to control the interaction with the radiation field.

\section{Relation to reduced density matrices} 
\label{subsection: Pauli relation to reduced density matrices}

The aim of this section is to show that condensation indicated by the counting functional, $\beta \rightarrow 0$ as $N \rightarrow \infty$, implies condensation in terms of reduced density matrices.

\lemma{
\label{lemma: Pauli relation between beta and reduced density matrices}
Let $\Psi_{N,t} \in \mathcal{D}(H_N)$, $\varphi_t \in L^2(\mathbb{R}^3)$ with  
$\norm{\varphi_t}=1$ and $u_t \in \mathfrak{h}$. Then
\begin{align}
\label{eq: Pauli relation between beta and reduced density matrices 1}
 \beta^a(\Psi_{N,t},\varphi_t) \leq&  \text{Tr}_{L^2(\mathbb{R}^3)} \abs{\gamma_{N,t}^{(1,0)} - \ket{\varphi_t} \bra{\varphi_t}} \leq  \sqrt{8  \beta^a(\Psi_{N,t},\varphi_t)}  ,  \\
\label{eq: Pauli relation between beta and reduced density matrices 2}
\text{Tr}_{\mathfrak{h}} \abs{\gamma_{N,t}^{(0,1)} - \ket{u_t}\bra{u_t}}  \leq&  3 \beta^b(\Psi_{Nt},u_t) + 6 \norm{u_t}_{\mathfrak{h}} \sqrt{\beta^b(\Psi_{Nt},u_t)}.
\end{align}
}

\proof{
The first inequality follows  from\footnote{For ease of notation, we discard the explicit time dependence and write for example $\Psi_N$ instead of $\Psi_{N,t}$.}
\begin{align}
\beta^a &= 1 - \SCP{\Psi_N}{p_1 \Psi_N}  =
1 - \scp{\varphi}{\gamma_N^{(1,0)} \varphi}
= \text{Tr}_{L^2(\mathbb{R}^3)} ( \ket{\varphi} \bra{\varphi} - \ket{\varphi} \bra{\varphi} \gamma_N^{(1,0)}) 
\nonumber \\
&\leq \norm{p_1}_{\op}  \text{Tr}_{L^2(\mathbb{R}^3)} \abs{\gamma_{N}^{(1,0)} - \ket{\varphi} \bra{\varphi}}
=  \text{Tr}_{L^2(\mathbb{R}^3)} \abs{\gamma_{N}^{(1,0)} - \ket{\varphi} \bra{\varphi}}.
\end{align}
In order to proof the remaining inequalities we use
\begin{align}
\label{eq: Pauli bound mit Hilbert Schmidt norm}
\text{Tr} \abs{\gamma - p} \leq 2 \norm{\gamma - p}_{HS}
+ \text{Tr} (\gamma - p),
\end{align}
valid for any one-dimensional projector $p$ and non-negative density matrix $\gamma$. The original argument of the proof  was first observed by Robert Seiringer, see \cite{rodnianskischlein}. We present a version that is found in \cite{anapolitanosmott}: 
Let $(\lambda_n)_{n \in \mathbb{N}}$ be the sequence of eigenvalues of the trace class operator $A \coloneqq \gamma - p$.
Since $p$ is a rank one projection, $A$ has at most one negative eigenvalue.
If there is no negative eigenvalue, $\text{Tr} \abs{A} = \text{Tr} (A)$ and \eqref{eq: Pauli bound mit Hilbert Schmidt norm} holds. If there is one negative eigenvalue $\lambda_1$, we have
$\text{Tr} \abs{A} = \abs{\lambda_1} + \sum_{n \geq 2} \lambda_n = 2 \abs{\lambda_1} + \text{Tr}(A). $
Because of $\abs{\lambda_1} \leq \norm{A}_{\op} \leq \norm{A}_{HS}$, inequality  \eqref{eq: Pauli bound mit Hilbert Schmidt norm} follows. \\
For the upper bound of \eqref{eq: Pauli relation between beta and reduced density matrices 1} we notice that $\text{Tr}_{L^2(\mathbb{R}^3)} (\gamma_N^{(1,0)} - \ket{\varphi} \bra{\varphi}) = 0$. Then, \eqref{eq: Pauli bound mit Hilbert Schmidt norm} reduces to  
\begin{align}
\text{Tr}_{L^2(\mathbb{R}^3)} \abs{\gamma_N^{(1,0)} - \ket{\varphi} \bra{\varphi}} \leq 2 \norm{\gamma_N^{(1,0)} - \ket{\varphi} \bra{\varphi}}_{HS}
\end{align}
 and \eqref{eq: Pauli relation between beta and reduced density matrices 1} follows from
\begin{align}
\text{Tr}_{L^2(\mathbb{R}^3)} ( \gamma_N^{(1,0)} - \ket{\varphi} \bra{\varphi} )^2
=& 1 - 2  \text{Tr}_{L^2(\mathbb{R}^3)}  (\ket{\varphi} \bra{\varphi} \gamma_N^{(1,0)} ) +  \text{Tr}_{L^2(\mathbb{R}^3)} ((\gamma_N^{(1,0)})^2)
\nonumber \\
\leq& 2 (1 - \text{Tr}_{L^2(\mathbb{R}^3)}  (\ket{\varphi} \bra{\varphi} \gamma_N^{(1,0)} ))
= 2 \beta^a.
\end{align}
To prove inequality \eqref{eq: Pauli relation between beta and reduced density matrices 2} it is useful to write the kernel of 
$\gamma_{N}^{(0,1)} - \ket{u}\bra{u}$ as
\begin{align}
(\gamma_{N}^{(0,1)} &- \ket{u}\bra{u})(k,\lambda,l,\mu) 
= \abs{k}^{1/2} \abs{l}^{1/2} \left(  N^{-1} \SCP{\Psi_N}{a^*(l,\mu) a(k,\lambda) \Psi_N} - \alpha^*(l,\mu) \alpha(k,\lambda)  \right) 
\nonumber \\
&= \abs{k}^{1/2} \abs{l}^{1/2}  \SCP{\left( N^{-1/2} a(l,\mu) - \alpha(l,\mu) \right)\Psi_N}{\left( N^{-1/2} a(k,\lambda) - \alpha(k,\lambda) \right)\Psi_N}  
\nonumber \\
&+ \abs{k}^{1/2} \abs{l}^{1/2}  \SCP{\left( N^{-1/2} a(l,\mu) - \alpha(l,\mu) \right)\Psi_N}{\Psi_N} \,   \alpha(k,\lambda)  
\nonumber \\
&+ \abs{k}^{1/2} \abs{l}^{1/2}  \SCP{\Psi_N}{\left( N^{-1/2} a(k,\lambda) - \alpha(k,\lambda) \right)\Psi_N}  \, \alpha^*(l,\mu).
\end{align}
Cauchy-Schwarz inequality gives
\begin{align}
\abs{(\gamma_{N}^{(0,1)} &- \ket{u}\bra{u})(k,\lambda,l,\mu)  }^2
\nonumber \\
&\leq \abs{k} \abs{l} \norm{\left( N^{-1/2} a(k,\lambda) - \alpha(k,\lambda) \right)\Psi_N}^2
  \norm{\left( N^{-1/2} a(l,\mu) - \alpha(l,\mu) \right)\Psi_N}^2
\nonumber \\
&+ \abs{k} \abs{l}  \norm{\left( N^{-1/2} a(k,\lambda) - \alpha(k,\lambda) \right)\Psi_N}^2  \abs{\alpha(l,\mu)}^2 
\nonumber \\
&+ \abs{k} \abs{l}   \norm{\left( N^{-1/2} a(l,\mu) - \alpha(l,\mu) \right)\Psi_N}^2  \abs{\alpha(k,\lambda)}^2
\end{align}
and 
\begin{align}
\norm{\gamma_{N}^{(0,1)} - \ket{u}\bra{u}}_{HS}^2 &=  \sum_{\lambda,\mu \in \{1,2\}^2} 
\int \int d^3k d^3l  \abs{(\gamma_{N}^{(0,1)} - \ket{u}\bra{u})(k,\lambda,l,\mu)}^2 
\nonumber \\
&\leq (\beta^b)^2 + 2 \norm{u}_{\mathfrak{h}}^2 \beta^b
\end{align}
follows.
Similarly, 
\begin{align}
\text{Tr}_{\mathfrak{h}} (\gamma_{N}^{(0,1)} - \ket{u}\bra{u})
&\leq \sum_{\lambda =1,2} \int d^3k 
\abs{(\gamma_{N}^{(0,1)} - \ket{u}\bra{u})(k,\lambda,k,\lambda)} 
\nonumber \\
&\leq  \sum_{\lambda = 1,2} \int d^3k  \abs{k} \norm{\left( N^{-1/2} a(k,\lambda) - \alpha(k,\lambda) \right)\Psi_N}^2  \nonumber \\
&+ 2  \sum_{\lambda = 1,2} \int d^3k \abs{u(k,\lambda)} \abs{k}^{1/2} \norm{\left( N^{-1/2} a(k,\lambda) - \alpha(k,\lambda) \right)\Psi_N}.
\end{align}
Applying Schwarz's inequality with respect to the scalar product of $\mathfrak{h}$ yields
\begin{align}
\text{Tr}_{\mathfrak{h}} (\gamma_{N}^{(0,1)} - \ket{\alpha}\bra{\alpha}) &\leq  \beta^b  
+ 2 \norm{u}_{\mathfrak{h}}
\Big(  \sum_{\lambda = 1,2} \int d^3k \abs{k}  \norm{\left( N^{-1/2} a(k,\lambda) - \alpha(k,\lambda) \right)\Psi_N}^2  \Big)^{1/2} 
\nonumber \\
&\leq \beta^b + 2 \norm{u}_{\mathfrak{h}} \sqrt{\beta^b}.
\end{align}
Monotonicity of the square root and \eqref{eq: Pauli bound mit Hilbert Schmidt norm} give rise to \eqref{eq: Pauli relation between beta and reduced density matrices 2}.
}

\section{Estimates on the time derivative}
\label{section: Pauli estimates on the time derivative}

In this section, we control the change of $\beta$ in time. To this end we separately estimate the time derivative of $\beta^a$ and $\beta^b$. The value of $\beta^c$ is constant in time because the energies of the many-body and the Maxwell-Schr\"odinger system are conserved quantities. To control the difference between the quantized and classical vector potential by the functional $\beta^b$ it is convenient to introduce their positive and negative frequency parts. 
\begin{align}
\vAp(x) &\coloneqq \sum_{\lambda=1,2} \int d^3k \, \tilde{\kappa}(k)
\frac{1}{\sqrt{2 \abs{k}}} \vep_{\lambda}(k) e^{ikx} a(k,\lambda),  
\nonumber \\
\vAm(x) &\coloneqq \sum_{\lambda=1,2} \int d^3k \, \tilde{\kappa}(k)
\frac{1}{\sqrt{2 \abs{k}}} \vep_{\lambda}(k) e^{-ikx} a^*(k,\lambda),  
\nonumber \\
\vApc(x,t) &\coloneqq \sum_{\lambda=1,2} \int d^3k \, \tilde{\kappa}(k) 
\frac{1}{\sqrt{2 \abs{k}}} \vep_{\lambda}(k) e^{ikx} \alpha_t(k,\lambda),  
\nonumber \\
\vAmc(x,t) &\coloneqq \sum_{\lambda=1,2} \int d^3k \, \tilde{\kappa}(k)
\frac{1}{\sqrt{2 \abs{k}}} \vep_{\lambda}(k) e^{-ikx} \alpha_t^*(k,\lambda).
\end{align}
Moreover, it is helpful to define the positive and negative frequency parts of the quantum mechanical and classical electric field.
\begin{align}
\vEp(x) &\coloneqq \sum_{\lambda=1,2} \int d^3k \, \tilde{\kappa}(k) \sqrt{\frac{\abs{k}}{2}} \vep_{\lambda}(k) i
 e^{ikx} a(k,\lambda),  
\nonumber \\
\vEm(x) &\coloneqq \sum_{\lambda=1,2} \int d^3k \, \tilde{\kappa}(k) \sqrt{\frac{\abs{k}}{2}} \vep_{\lambda}(k) (-i)
 e^{-ikx} a^*(k,\lambda),  
\nonumber \\
 \label{eq: Pauli positive frequency part of the classical electric field}
\vEpc(x,t) &\coloneqq \sum_{\lambda=1,2} \int d^3k \, \tilde{\kappa}(k)   \sqrt{\frac{\abs{k}}{2}}
 \vep_{\lambda}(k) i e^{ikx} \alpha_t(k,\lambda),  
 \nonumber \\
\vEmc(x,t) &\coloneqq \sum_{\lambda=1,2} \int d^3k \, \tilde{\kappa}(k)  \sqrt{\frac{\abs{k}}{2}} \vep_{\lambda}(k) (- i) e^{-ikx} \alpha_t^*(k,\lambda).
\end{align}
For $\sharp \in \{ \ , +,- \}$, we introduce the shorthand notations
\begin{align}
\label{eq: Pauli shorthand notation field differences}
\vEd^{\sharp}(x,t)  &\coloneqq \frac{\vE^{\sharp}(x)}{\sqrt{N}} - \vEc^{\sharp}(x,t), \quad
\vAd^{\sharp}(x,t)  \coloneqq \frac{\vA^{\sharp}(x)}{\sqrt{N}} - \vAc ^{\sharp}(x,t) .
\end{align}
By means of the cutoff function 
\begin{equation}
\label{eq: Pauli cut off function 2}
\tilde{\eta}(k) \coloneqq \abs{k}^{-1}  \tilde{\kappa}(k) = 
  (2 \pi)^{- \frac{3}{2}} \abs{k}^{-1} \id_{\abs{k}\leq \Lambda}(k)
\end{equation}
we are able to express the vector potential in terms of the electric field.
\begin{lemma}
\label{lemma: Pauli auxiliary fields}
Let $\eta$ be the Fourier transform of \eqref{eq: Pauli cut off function 2}, then 
\begin{align}
\label{eq: Pauli relation A and E}
\vAp(x) &= - i \big( \eta * \vEp \big)(x), 
\qquad  \quad \; \;
\vAm(x) = i  \big( \eta * \vEm \big)(x), 
\nonumber \\
\vApc(x,t) &= - i \big( \eta * \vEpc \big)(x,t) ,
\qquad
\vAmc(x,t) =  i \big( \eta * \vEmc \big)(x,t) .
\end{align}
\end{lemma}
\begin{proof}
The proof is a simple application of the convolution theorem.
\end{proof}

\noindent
At various points in our estimates, we replace the vector potential by the electric field and make use of (see Lemma~\ref{lemma: Pauli field operator times p or q against functional})
\begin{align}
\int d^3y \, \SCP{\Psi_{N,t}}{\left( N^{-1/2} \vEm(y) - \vEmc(y,t) \right) \left( N^{-1/2} \vEp(y) - \vEpc(y,t) \right) \Psi_{N,t}}    &\leq \beta^b(t).
\end{align}
To obtain proper bounds it is crucial that the $L^2$-norm of the cutoff functions
\begin{align}
\label{eq: Pauli norm cut off function 2}
\norm{\kappa}_2^2 = \Lambda^3/(6 \pi^2) 
\quad \text{and} \quad
\norm{\eta}_2^2 = \Lambda/(2 \pi^2) 
\end{align}
is finite.

\subsection{Preliminary estimates }
\label{subsection: Pauli preliminary estimates}
The minimal coupling term in the Pauli-Fierz Hamiltonian
\begin{align}
\sum_{j=1}^N (- i \nabla_j - N^{-1/2} \vA(x_j))^2 &= \sum_{j=1}^N \left(   - \Delta_1 + 2 i N^{-1/2} \vA(x_1) \cdot \nabla_1 + N^{-1} \vA^2(x_1) \right)
\end{align}
contains an interaction that is quadratic in the vector potential. If we want to control the growth of $\beta(t)$ in time this quadratic part (see for example the estimate of \eqref{eq: Pauli alpha-a-field}) requires that quantities like $\norm{\nabla_2 q_1 \Psi_{N,t}}^2$ and $N^{-1} \norm{\vA(x_1) q_1 \Psi_{N,t}}^2$ are not only finite but bounded by $\beta(t)$. This holds for  every bounded operator $B$ because of
\begin{align}
\SCP{\Psi_{N,t}}{q_1 B q_1 \Psi_{N,t}} \leq C \norm{q_1 \Psi_{N,t}}^2 \leq C \beta^a(t)
\end{align}
but must not be true in general. In case of unbounded operators smallness can sometimes be inferred on a subclass of states which have sufficient decay in the occupation of eigenstates with large eigenvalues. 
For a self-adjoint operator $O$ with $\left[O, q_1 \right] \approx 0$ and $c \in \mathbb{R}$ one has
\begin{align}
\SCP{\Psi_{N,t}}{q_1 O q_1 \Psi_{N,t}}  
\approx& \SCP{\Psi_{N,t}}{q_1 O \Psi_{N,t}}
= \SCP{\Psi_{N,t}}{q_1 \left( O - c \right) \Psi_{N,t}} + c \SCP{\Psi_{N,t}}{q_1 \Psi_{N,t}} 
\nonumber \\
\leq& (c+1) \SCP{\Psi_{N,t}}{q_1 \Psi_{N,t}} 
+  \SCP{\Psi_{N,t}}{\left( O - c \right)^2 \Psi_{N,t}}.
\end{align}
Thus,  $\SCP{\Psi_{N,t}}{q_1 O q_1 \Psi_{N,t}}$ can be bounded by $\beta(t)$ plus a small error if  $\Psi_{N,t}$ occupies eigenstates of $O$ with eigenvalues $\lambda \neq c$ only with small probability. This is in the spirit of Chebyshev's inequality from probability theory. Requiring $\SCP{\Psi_{N,0}}{\left( O - c \right)^2 \Psi_{N,0}} \approx 0$ initially does not imply smallness at later times. However, if $O$ is a conserved quantity its variance is constant during the time evolution and we only have to restrict the class of initial states. In the following, we consider the variance of the energy per particles of the many-body system (see $\beta^c$). Then, we estimate the vector potential and the Laplacian by $H_N/N$ and bound expression like $N^{-1} \SCP{\Psi_{N,t}}{q_1 \vA^2(x_1) q_1 \Psi_{N,t}}$ by the 
counting functional.

\begin{lemma}
\label{lemma: Pauli Feldoperator-Abschaetzung}
Let $y \in \mathbb{R}^3$ or $y \in \{x_1, \ldots, x_N \}$ and $\Psi_N \in  \mathcal{D}(H_N)$. Then
\begin{align}
\norm{ N^{-1/2} \vAp(y) \Psi_N}^2
&\leq \Lambda /(2 \pi^2)  \SCP{\Psi_N}{N^{-1} H_f \Psi_N} , 
\nonumber \\
\norm{ N^{-1/2} \vAm(y) \Psi_N}^2 
&\leq \Lambda/(2 \pi^2) \SCP{\Psi_N}{N^{-1} H_f \Psi_N}
+  \Lambda^2 / (4 \pi^2 N) \norm{\Psi_N}^2,
\nonumber \\
\norm{ N^{-1/2} \vA(y) \Psi_N}^2 
&\leq 2 \Lambda / \pi^2 \SCP{\Psi_N}{ N^{-1} H_f \Psi_N}
+ \Lambda^2 / (2 \pi^2 N) \norm{\Psi_N}^2.
\end{align}
\end{lemma}

\begin{proof}
To ease notation, we define the vector-valued function 
$\vf(k, \lambda) \coloneqq  \frac{\tilde{\kappa}(k)}{\sqrt{2 \abs{k}}}   \vep_{\lambda}(k)$. 
The first estimate follows from Cauchy-Schwarz inequality
\begin{align}
\Big| \Big| &\sum_{\lambda=1,2} \int d^3k \, \vf(k,\lambda) e^{\pm i k y} a(k,\lambda) \Psi_N \Big| \Big|^2   
\nonumber \\
&\leq \Big( \sum_{\lambda=1,2} \int d^3k \, \abs{\vf(k,\lambda)} \abs{k}^{-1/2} \norm{\abs{k}^{1/2} a(k,\lambda) \Psi_N}^2  \Big)^2
\nonumber \\
&\leq  \Big( \sum_{\lambda=1,2} \int d^3k \,  \abs{\vf(k,\lambda)}^2 \abs{k}^{-1} \Big)
\Big(  \sum_{\lambda=1,2} \int d^3k \, \abs{k} \norm{ a(k,\lambda) \Psi_N}^2 \Big) 
\nonumber \\ 
&= \Lambda/(2 \pi^2) \SCP{\Psi_N}{H_f \Psi_N}.
\end{align}
By use of the canonical commutation relations~\eqref{eq: canonical commutation relation}, the second bound is obtained via
\begin{align}
\Big| \Big| &\sum_{\lambda=1,2} \int d^3k \, \vf(k,\lambda) e^{\pm i k y} a^*(k,\lambda) \Psi_N \Big| \Big|^2 
=   \Big| \Big|\sum_{\lambda=1,2} \int d^3k \, \overline{\vf(k,\lambda)} e^{\mp ik y} a(k,\lambda) \Psi_N \Big| \Big|^2  
\nonumber \\
&+\norm{\vf}_{\mathfrak{h}}^2 \norm{\Psi_N}^2  \leq  \Lambda^2/(4 \pi^2) \norm{\Psi_N}^2
+  \Lambda/(2 \pi^2) \SCP{\Psi_N}{H_f \Psi_N}.
\end{align}
The last estimate follows by triangular inequality.
\end{proof}

\noindent
Lemma~\ref{lemma: Pauli Feldoperator-Abschaetzung} leads to

\begin{corollary}
\label{corollary: Pauli Feldoperator-Abschaetzung mit q1}
For $y \in \mathbb{R}^3$ or $y \in \{x_1, \ldots, x_N \}$ and $\Psi_N \in \mathcal{D}(H_N)$ we have
\begin{align}
\norm{N^{-1/2}\vA(y) q_1 \Psi_N}^2 
\leq&2 \Lambda/\pi^2 \SCP{\Psi_N}{q_1 N^{-1} H_f q_1 \Psi_N}
+ \Lambda^2/(2 \pi^2 N) \beta^a,  \nonumber \\
\norm{N^{-1/2} \vA(y) p_1 q_2 \Psi_N}^2
\leq& 2 \Lambda/\pi^2 \SCP{\Psi_N}{q_1 N^{-1} H_f q_1 \Psi_N} + \Lambda^2/(2 \pi^2 N) \beta^a.
\end{align}

\end{corollary}

\begin{lemma}
\label{lemma: Pauli commutator of H with q1}
Let $v$  satisfy (A1), $\Psi_{N,t} \in \left(L^2_s \left( \mathbb{R}^{3N} \right) \otimes \mathcal{F}_p \right) \cap \mathcal{D}(H_N)$ and $(\varphi_t,\boldsymbol{A}(t), \boldsymbol{E}(t)) \in \mathcal{G}$. Then,
there exists a monotone increasing function $C(t)$ of $\mathcal{E}_M[\varphi_t,u_t]$, $\norm{\varphi_t}_{H^2(\mathbb{R}^3)}$ and $\norm{\varphi_t}_{L^{\infty}(\mathbb{R}^3)}$  such that
\begin{align}
\SCP{\Psi_{N,t}}{q^{\varphi_t}_1 N^{-1} H_N q^{\varphi_t}_1 \Psi_{N,t}}
\leq& C(t) \left( \beta(t) + \Lambda/N \right).
\end{align}
\end{lemma}

\begin{proof}
We decompose the Pauli-Fierz Hamiltonian into
\begin{align}
\label{eq: commutator of H with q1-1}
\SCP{\Psi_N}{q_1 N^{-1} H_N q_1 \Psi_N}
=& \SCP{\Psi_N}{q_1 N^{-1} \sum_{j=1}^N 
 \left( - i  \nabla_j -  N^{-1/2} \vA(x_j)   \right)^2 q_1 \Psi_N}
\\
\label{eq: commutator of H with q1-3}
+& \SCP{\Psi_N}{q_1 N^{-2} \sum_{1 \leq j < k \leq N} v(x_j - x_k) q_1 \Psi_N}
\\
\label{eq: commutator of H with q1-4}
+& \SCP{\Psi_N}{q_1 N^{-1}  H_f q_1 \Psi_N}.
\end{align}
Then, we write the first line as
\begin{align}
\eqref{eq: commutator of H with q1-1}
&= \SCP{\Psi_N}{q_1 N^{-1} \sum_{j=1}^N 
 \left( - i  \nabla_j - N^{-1/2}  \vA(x_j) \right)^2  \Psi_N}   
\nonumber  \\
 &+ N^{-1} \sum_{j=1}^N 
 \SCP{\Psi_N}{q_1 \left[ \left( - i  \nabla_j -   N^{-1/2}  \vA(x_j) \right)^2 , q_1  \right] \Psi_N} 
 \nonumber \\
&= \SCP{\Psi_N}{q_1 N^{-1} \sum_{j=1}^N 
 \left( - i  \nabla_j -   N^{-1/2}  \vA(x_j)   \right)^2  \Psi_N}  
 \nonumber \\
&+ N^{-1}
 \SCP{\Psi_N}{q_1 \left[ \left( - i  \nabla_1 -   N^{-1/2} \vA(x_1) \right)^2 , q_1  \right] \Psi_N}  .
\end{align}
The second line is given by
\begin{align}
\eqref{eq: commutator of H with q1-3}
&= \SCP{\Psi_N}{q_1 N^{-2} \sum_{1 \leq j < k \leq N} v(x_j - x_k)  \Psi_N} 
+ N^{-2}   \sum_{1 \leq j < k \leq N}
 \SCP{\Psi_N}{q_1 \left[  v(x_j - x_k), q_1 \right] \Psi_N} 
\nonumber \\
&= \SCP{\Psi_N}{q_1 N^{-2} \sum_{1 \leq j < k \leq N} v(x_j - x_k)  \Psi_N}  
+ (N-1) N^{-2}  \SCP{\Psi_N}{q_1 \left[ v(x_1 - x_2) , q_1 \right] \Psi_N}  .
\end{align}
In line~\eqref{eq: commutator of H with q1-4} we use that $H_f$ commutes with operators which only act on the sector of the non-relativistic particles. This leads to
\begin{align}
\SCP{\Psi_N}{q_1 N^{-1} H_N q_1 \Psi_N}
&= \SCP{\Psi_N}{q_1 N^{-1} H_N \Psi_N} 
\nonumber \\
\label{eq: commutator of H with q1-5}
&+ N^{-1} \SCP{\Psi_N}{q_1\left[ \left( - i \nabla_1 -  N^{-1/2} \vA(x_1) \right)^2,q_1 \right] \Psi_N}  \\
\label{eq: commutator of H with q1-7}
&+ (N-1) N^{-2} \SCP{\Psi_N}{q_1 \left[ v(x_1 - x_2) , q_1  \right] \Psi_N}  .
\end{align}
The first term is estimated by
\begin{align}
\abs{\eqref{eq: commutator of H with q1-5}}
&=  N^{-1}  \abs{\SCP{\Psi_N}{q_1\left[ \left( - i  \nabla_1 -  N^{-1/2} \vA(x_1) \right)^2,p_1 \right] \Psi_N}}   
\nonumber \\
&= N^{-1}  \abs{\SCP{\Psi_N}{q_1 \left( - i \nabla_1 - N^{-1/2} \vA(x_1) \right)^2 p_1  \Psi_N}}  
\nonumber \\
&\leq  N^{-1} \abs{\SCP{q_1 \Psi_N}{ (- \Delta_1) p_1 \Psi_N}}
\nonumber \\
&+ N^{-1} \abs{\SCP{N^{-1/2} \vA(x_1) q_1 \Psi_N}{\nabla_1 p_1 \Psi_N}}
\nonumber\\
&+ N^{-1} \abs{\SCP{N^{-1/2} \vA(x_1) q_1 \Psi_N}{N^{-1/2} \vA(x_1) p_1 \Psi_N}}
\nonumber\\
&\leq N^{-1} \left(   \beta^a + \norm{\Delta p_1 \Psi_N}^2 + \norm{\nabla p_1 \Psi_N}^2 \right)
\nonumber \\
&+ N^{-1} \left(  \norm{N^{-1/2} \vA(x_1) q_1 \Psi_N}^2 + \norm{N^{-1/2} \vA(x_1) p_1 \Psi_N}^2   \right).
\end{align}
Lemma~\ref{lemma: Pauli Feldoperator-Abschaetzung} and the positivity of the interaction potential $v$ let us continue with
\begin{align}
\abs{\eqref{eq: commutator of H with q1-5}}
&\leq N^{-1}  \Lambda C(\norm{\varphi}_{H^2})  \left(  \SCP{\Psi_N}{N^{-1} H_f \Psi_N} 
+ \Lambda/N \right)
\nonumber \\
&\leq N^{-1} \Lambda C(\norm{\varphi}_{H^2}) \left( \SCP{\Psi_N}{N^{-1} H_N \Psi_N} + \Lambda/N \right)
\nonumber\\
&\leq N^{-1} \Lambda C(\norm{\varphi}_{H^2}) 
\left( \SCP{\Psi_N}{ (N^{-1} H_N - \mathcal{E}_M ) \Psi_N} + \mathcal{E}_M \right)
\nonumber\\
&\leq  N^{-1} \Lambda C(\norm{\varphi}_{H^2}) 
\left(  \norm{(N^{-1} H_N - \mathcal{E}_M ) \Psi_N} + \mathcal{E}_M  \right)
\nonumber\\
&\leq  N^{-1} \Lambda C(\norm{\varphi}_{H^2})
\left( \sqrt{\beta^c} + \mathcal{E}_M \right) 
\leq N^{-1} \Lambda C(\norm{\varphi}_{H^2}, \mathcal{E}_M). 
\end{align}
The second term is bounded by
\begin{align}
\abs{\eqref{eq: commutator of H with q1-7}}
&\leq N^{-1} \abs{\SCP{\Psi_N}{q_1 \left[v(x_1 - x_2) , p_1 \right] \Psi_N}}  
= N^{-1} \abs{\SCP{\Psi_N}{q_1 v(x_1 - x_2)  p_1  \Psi_N}}   
\nonumber \\
&\leq 1/2 \norm{q_1 \Psi_N}^2 + 1/(2 N^2)
\norm{v(x_1 - x_1) p_1 \Psi_N}^2  
\nonumber \\
&= 1/2 \beta^a  
+ 1/(2N^2)  \SCP{\Psi_N}{p_1 v^2(x_2 - x_1) p_1 \Psi_N} 
\nonumber \\
&= 1/2 \beta^a  
+ 1/(2N^2)  \SCP{\Psi_N}{p_1 \left( v^2 * \abs{\varphi}^2 \right)(x_2) \Psi_N} 
\nonumber \\
&\leq 1/2 \beta^a + 1/(2 N^2) \norm{v^2 * \abs{\varphi}^2}_{\infty}.
\end{align}
We use assumption (A1) and decompose the interaction potential $v = v_1 + v_2$ into $v_1 \in L^2(\mathbb{R}^3)$ and $v_2 \in L^{\infty}(\mathbb{R}^3)$. Then, we apply Young's inequality and obtain
\begin{align}
\label{eq: Pauli direct interaction potential Young 1}
\norm{v^2 * \abs{\varphi}^2 }_{\infty} &\leq \norm{v_1^2}_1 \norm{\abs{\varphi}^2}_{\infty} + \norm{v_2^2}_{\infty} \norm{\abs{\varphi}^2}_1
\nonumber \\
&= \norm{v_1}_2^2 \norm{\varphi}_{\infty}^2 + \norm{v_2}_{\infty}^2 \norm{\varphi}_2^2
\leq C(\norm{\varphi}_{\infty}).
\end{align}
Thus,
\begin{align}
\abs{\eqref{eq: commutator of H with q1-5}
+ \eqref{eq: commutator of H with q1-7}}
\leq C(\norm{\varphi}_{H^2}, \norm{\varphi}_{\infty},\mathcal{E}_M) \left(  \beta + \Lambda/N \right)
\end{align}
and 
\begin{align}
\SCP{\Psi_N}{&q_1 N^{-1} H_N q_1 \Psi_N}
\leq  \abs{\SCP{\Psi_N}{q_1 N^{-1} H_N \Psi_N}} 
+ \abs{\eqref{eq: commutator of H with q1-5}
+ \eqref{eq: commutator of H with q1-7}}  
\nonumber \\
&\leq  \abs{\SCP{\Psi_N}{q_1 \left( N^{-1} H_N - \mathcal{E}_M \right)   \Psi_N} + \mathcal{E}_M \beta^a} 
+ C(\norm{\varphi}_{H^2}, \norm{\varphi}_{\infty},\mathcal{E}_M) \left(  \beta + \Lambda/N \right) 
\nonumber \\
&\leq \abs{\SCP{\Psi_N}{q_1 \left( N^{-1} H_N - \mathcal{E}_M \right)   \Psi_N}}  
+ C(\norm{\varphi}_{H^2}, \norm{\varphi}_{\infty},\mathcal{E}_M) \left(  \beta + \Lambda/N \right) 
\nonumber \\
&\leq  \norm{\left( N^{-1} H_N - \mathcal{E}_M \right) \Psi_N}^2 +  \beta^a 
+ C(\norm{\varphi}_{H^2}, \norm{\varphi}_{\infty},\mathcal{E}_M) \left(  \beta + \Lambda/N \right)  
\nonumber \\
&\leq C(\norm{\varphi}_{H^2}, \norm{\varphi}_{\infty},\mathcal{E}_M) \left(  \beta + \Lambda/N \right).
\end{align}

\end{proof}

\begin{lemma}
\label{lemma: Pauli Feldoperator gegen alpha Abschaetzung}
Let $y \in \mathbb{R}^3$ or $y \in \{x_1, \ldots, x_N \}$, $v$  satisfy (A1), $\Psi_{N,t} \in \left(L^2_s \left( \mathbb{R}^{3N} \right) \otimes \mathcal{F}_p \right) \cap \mathcal{D}(H_N)$ and $(\varphi_t,\boldsymbol{A}(t), \boldsymbol{E}(t)) \in \mathcal{G}$. Then,
there exists a monotone increasing function $C(t)$ of $\mathcal{E}_M[\varphi_t,u_t]$, $\norm{\varphi_t}_{H^2(\mathbb{R}^3)}$ and $\norm{\varphi_t}_{L^{\infty}(\mathbb{R}^3)}$  such that
\begin{align}
\norm{ N^{-1/2} \vA(y) q_1 \Psi_{N,t}}^2
&\leq  \Lambda C(t)  \left( \beta(t) + \Lambda/N \right),  
\nonumber \\
\norm{ N^{-1/2} \vA(y) q_2 \Psi_{N,t}}^2
&\leq  \Lambda C(t)  \left( \beta(t) + \Lambda/N \right),  
\nonumber \\
\norm{ N^{-1/2} \vA(y) p_1 q_2 \Psi_N}^2
&\leq  \Lambda C(t)  \left( \beta(t) + \Lambda/N \right).
\end{align}
\end{lemma}

\begin{proof}
We have
\begin{align}
\SCP{\Psi_N}{q_1 N^{-1} H_f q_1 \Psi_N}
\leq& \SCP{\Psi_N}{q_1 N^{-1} H_N q_1 \Psi_N} 
\end{align}
because $v$ is positive.
Lemma~\ref{lemma: Pauli commutator of H with q1} and Corollary~\ref{corollary: Pauli Feldoperator-Abschaetzung mit q1} then lead to
\begin{align}
\SCP{\Psi_N}{q_1 N^{-1} H_f q_1 \Psi_N}
\leq& C(\norm{\varphi}_{H^2}, \norm{\varphi}_{\infty},\mathcal{E}_M) \left( \beta + \Lambda/N   \right)
\end{align}
and
\begin{align}
\norm{ N^{-1/2} \vA(y) q_1 \Psi_N}^2 
\leq& \Lambda  C(\norm{\varphi}_{H^2}, \norm{\varphi}_{\infty},\mathcal{E}_M)  \left( \beta + \Lambda/N \right).
\end{align}
The other inequalities are shown analogously.
\end{proof}

\begin{lemma}
\label{lemma: crucial bound}
Let $v$  satisfy (A1), $\Psi_{N,t} \in \left(L^2_s \left( \mathbb{R}^{3N} \right) \otimes \mathcal{F}_p \right) \cap \mathcal{D}(H_N)$ and $(\varphi_t,\boldsymbol{A}(t), \boldsymbol{E}(t)) \in \mathcal{G}$. Then, there exists a monotone increasing function $C(t)$ of $\mathcal{E}_M[\varphi_t,u_t]$, $\norm{\varphi_t}_{H^2(\mathbb{R}^3)}$ and $\norm{\varphi_t}_{L^{\infty}(\mathbb{R}^3)}$  such that
\begin{align}
 \int d^3y \, \Big| \Big| N^{-1} \sum_{j=1}^N q_j \kappa(x_j -y) \left( - i  \nabla_j - N^{-1/2} \vA(x_j) \right) \Psi_{N,t} \Big| \Big|^2   
&\leq  \Lambda^3 C(t) \left(  \beta + \Lambda/N \right) .
\end{align}
\end{lemma}

\begin{proof}
In the following, we use  $\SCP{\cdot}{\cdot}_{;y} = \int d^3y \, \SCP{\cdot}{\cdot}$ and $\norm{\cdot}_{;y} = \sqrt{\int d^3y \, \SCP{\cdot}{\cdot}}$ to ease the notation. We estimate
\small
\begin{align}
& \int d^3y \, \Big| \Big| N^{-1} \sum_{j=1}^N q_j \kappa(x_j -y) \left( - i  \nabla_j - N^{-1/2} \vA(x_j) \right) \Psi_{N} \Big| \Big|^2   
\nonumber\\
&= N^{-2} \SCP{\sum_{i=1}^N q_i \kappa(x_i -y) \left( - i  \nabla_i - N^{-1/2} \vA(x_i) \right) \Psi_{N}}{\sum_{j=1}^N q_j \kappa(x_j -y) \left( - i  \nabla_j - N^{-1/2} \vA(x_j) \right) \Psi_{N}}_{;y}
\nonumber\\
&= N^{-1} \SCP{ q_1 \kappa(x_1 -y) \left( - i  \nabla_1 - N^{-1/2} \vA(x_1) \right) \Psi_{N}}{q_1 \kappa(x_1 -y) \left( - i  \nabla_1 - N^{-1/2} \vA(x_1) \right) \Psi_{N}}_{;y}
\nonumber\\
&+ (N-1) N^{-1} \SCP{ q_1 \kappa(x_1 -y) \left( - i  \nabla_1 - N^{-1/2} \vA(x_1) \right) \Psi_{N}}{q_2 \kappa(x_2 -y) \left( - i  \nabla_2 - N^{-1/2} \vA(x_2) \right) \Psi_{N}}_{;y}
\nonumber\\
&\leq N^{-1} \norm{\kappa(x_1 -y) \left( - i  \nabla_1 - N^{-1/2} \vA(x_1) \right) \Psi_{N}}_{;y}^2
\nonumber \\
&+ (N-1) N^{-1} \SCP{  \kappa(x_1 -y) \left( - i  \nabla_1 - N^{-1/2} \vA(x_1) \right) q_2 \Psi_{N}}{\kappa(x_2 -y) \left( - i  \nabla_2 - N^{-1/2} \vA(x_2) \right)  q_1 \Psi_{N}}_{;y}
\nonumber\\
&\leq  N^{-1} \norm{\kappa(x_1 -y) \left( - i  \nabla_1 - N^{-1/2} \vA(x_1) \right) \Psi_{N}}_{;y}^2
\nonumber\\
&+ (N-1) N^{-1} \norm{\kappa(x_1 -y) \left( - i  \nabla_1 - N^{-1/2} \vA(x_1) \right) q_2 \Psi_{N}}^2_{;y} 
\nonumber \\
&= N^{-1} \SCP{\left( - i  \nabla_1 - N^{-1/2} \vA(x_1) \right) \Psi_{N}}{ \Big( \int d^3y \, \abs{\kappa(x_1 -y)}^2 \Big) \left( - i  \nabla_1 - N^{-1/2} \vA(x_1) \right) \Psi_{N}}
\nonumber \\
&+ (N-1) N^{-1}  \SCP{\left( - i  \nabla_1 - N^{-1/2} \vA(x_1) \right) q_2 \Psi_{N}}{ \Big( \int d^3y \, \abs{\kappa(x_1 -y)}^2 \Big) \left( - i  \nabla_1 - N^{-1/2} \vA(x_1) \right) q_2 \Psi_{N}}
\nonumber \\
&= N^{-1} \norm{\kappa}_2^2 \SCP{  \Psi_N}{ \left( - i  \nabla_1 - N^{-1/2} \vA(x_1) \right)^2 \Psi_N} 
\nonumber \\
&+ (N-1) N^{-1} \norm{\kappa}_2^2  \SCP{  \Psi_N}{q_2 \left( - i  \nabla_1 -  N^{-1/2} \vA(x_1) \right)^2 q_2 \Psi_N}.
\end{align}
\normalsize
So if we insert the identity $1 = p_1 + q_1$ and use the symmetry of the wave function, we get
\begin{align}
 \int d^3y \, \Big| \Big| N^{-1} &\sum_{j=1}^N q_j \kappa(x_j -y) \left( - i  \nabla_j - N^{-1/2} \vA(x_j) \right) \Psi_{N} \Big| \Big|^2   
 \nonumber\\
&\leq N^{-1} \norm{\kappa}_2^2  \SCP{  \Psi_N}{q_1 \left( - i  \nabla_1 -   N^{-1/2} \vA(x_1) \right)^2 q_1 \Psi_N}  
\nonumber \\
&+ 2 N^{-1} \norm{\kappa}_2^2  \abs{\SCP{  \Psi_N}{q_1 \left( - i  \nabla_1 -   N^{-1/2} \vA(x_1) \right)^2 p_1 \Psi_N}} 
\nonumber \\
&+ N^{-1} \norm{\kappa}_2^2  \abs{\SCP{  \Psi_N}{p_1 \left( - i  \nabla_1 -   N^{-1/2} \vA(x_1) \right)^2 p_1 \Psi_N}}  
\nonumber \\
&+ N^{-1} \norm{\kappa}_2^2 \sum_{j=2}^N
\SCP{\Psi_N}{q_1   \left( - i  \nabla_j -  N^{-1/2} \vA(x_j) \right)^2 q_1 \Psi_N}.
\end{align}
Adding the lines together this simplifies to
\begin{align}
 \int d^3y \, \Big| \Big| N^{-1} &\sum_{j=1}^N q_j \kappa(x_j -y) \left( - i  \nabla_j - N^{-1/2} \vA(x_j) \right) \Psi_{N} \Big| \Big|^2   
 \nonumber\\
&= N^{-1}  \norm{\kappa}_2^2 \sum_{j=1}^N
\SCP{\Psi_N}{q_1   \left( - i  \nabla_j -  N^{1/2} \vA(x_j) \right)^2 q_1 \Psi_N}  
\nonumber \\
&+2 N^{-1} \norm{\kappa}_2^2  \abs{\SCP{  \Psi_N}{q_1 \left( - i  \nabla_1 - N^{-1/2} \vA(x_1) \right)^2 p_1 \Psi_N}} 
\nonumber \\
&+  N^{-1} \norm{\kappa}_2^2  \abs{\SCP{  \Psi_N}{p_1 \left( - i  \nabla_1 - N^{-1/2} \vA(x_1) \right)^2 p_1 \Psi_N}}.
\end{align}  
Now we estimate the last two lines analogously to~\eqref{eq: commutator of H with q1-5} and obtain
\begin{align} 
 \int d^3y \, \Big| \Big| N^{-1} &\sum_{j=1}^N q_j \kappa(x_j -y) \left( - i  \nabla_j - N^{-1/2} \vA(x_j) \right) \Psi_{N} \Big| \Big|^2   
\leq  \norm{\kappa}_2^2  \Lambda/N C(\norm{\varphi}_{H^2},\mathcal{E}_M)
\nonumber \\
&+  N^{-1}  \norm{\kappa}_2^2 \sum_{j=1}^N
\SCP{\Psi_N}{q_1   \left( - i  \nabla_j -  N^{-1/2} \vA(x_j) \right)^2 q_1 \Psi_N}.
\end{align}
Because $H_f$ and $v$ are positive operators, this is bounded by
\begin{align}
 \norm{\kappa}_2^2 \SCP{\Psi_N}{q_1 N^{-1} H_N q_1 \Psi_N}     
+  \norm{\kappa}_2^2  \Lambda/N  C(\norm{\varphi}_{H^2},\mathcal{E}_M).
\end{align}
Then, we apply Lemma~\ref{lemma: Pauli commutator of H with q1} and obtain
\begin{align}
 \int d^3y \, \Big| \Big| N^{-1} \sum_{j=1}^N q_j \kappa(x_j -y) \left( - i  \nabla_j - N^{-1/2} \vA(x_j) \right) \Psi_{N} \Big| \Big|^2   
&\leq    \norm{\kappa}_2^2 C(t) \left( \beta + \Lambda/N \right) 
\nonumber \\
&\leq  \Lambda^3 C(t)  \left( \beta + \Lambda/N \right),
\end{align}
where $C(t)$ is a monotone increasing function of $\mathcal{E}_M[\varphi_t,u_t]$, $\norm{\varphi_t}_{H^2(\mathbb{R}^3)}$ and $\norm{\varphi_t}_{L^{\infty}(\mathbb{R}^3)}$.
\end{proof}

\subsection{Bound on $d_t \beta^a$:}
\lemma{
\label{lemma: dt beta-a}
Let $v$ satisfy (A1), $(\varphi_t,\boldsymbol{A}(t), \boldsymbol{E}(t)) \in \mathcal{G}$ and  $\Psi_{N,t}$ be the unique solution of \eqref{eq: Pauli Schroedinger equation microscopic} with initial data $\Psi_{N,0} \in \left(L^2_s \left( \mathbb{R}^{3N} \right) \otimes \mathcal{F}_p \right) \cap \mathcal{D}(H_N)$. 
Then, there exists a monotone increasing function $C(t)$ of $\norm{\vAc}_{\infty}$, $\mathcal{E}_M[\varphi_t,u_t]$, $\norm{\varphi_t}_{H^2(\mathbb{R}^3)}$ and $\norm{\varphi_t}_{L^{\infty}(\mathbb{R}^3)}$ such that
\begin{align}
\abs{d_t \beta^a(\Psi_{N,t},\varphi_t)} \leq  \Lambda^2 C(t)  \left( \beta(\Psi_{N,t},\varphi_t, u_t) + \Lambda/N \right).
\end{align}
}
\proof{ The time derivative of the projector $q^{\varphi_t}_1$ is given by
\begin{align}
d_t q^{\varphi_t}_1 &= - i \left[ H_1^{HM}, q_1^{\varphi_t} \right],
\end{align}
where $H_1^{HM}$ denotes the effective Hamiltonian $H_1^{HM} \coloneqq \left( - i \nabla_1 - \vAc(x_1,t) \right)^2 + \left(v * \abs{\varphi_t}^2 \right)(x_1)$.
This allows us to compute the derivative of $\beta^a(t)$ by
\begin{align}
d_t \beta^a(t) &=  d_t \SCP{\Psi_{N,t}}{q^{\varphi_t}_1 \Psi_{N,t}} 
=  i  \SCP{\Psi_{N,t}}{\left[ \left(H_N - H_1^{HM} \right), q^{\varphi_t}_1 \right] \Psi_{N,t}} 
\nonumber \\
&=  - 2 \SCP{\Psi_{N,t}}{\left[ \left(  N^{-1/2} \vA(x_1) - \vAc(x_1,t) \right) \cdot \nabla_1, q^{\varphi_t}_1 \right] \Psi_{N,t}} 
\nonumber \\
&  \; \; \; \,  +  i  \SCP{\Psi_{N,t}}{\left[ \left( N^{-1} \vA^2(x_1) - \vAc^2(x_1,t) \right) , q^{\varphi_t}_1 \right] \Psi_{N,t}} 
\nonumber \\
&  \; \; \; \,  + i  \SCP{\Psi_{N,t}}{\Big[ \Big( N^{-1} \sum_{1\leq j < k \leq N} v(x_j - x_k)  -  \left(v * \abs{\varphi_t}^2 \right)(x_1) \Big)  , q^{\varphi_t}_1 \Big] \Psi_{N,t}}   
\nonumber \\
&= - 4  \Re \SCP{\Psi_{N,t}}{\left(  N^{-1/2} \vA(x_1) - \vAc(x_1,t) \right) \cdot \nabla_1 q^{\varphi_t}_1 \Psi_{N,t}} 
\nonumber \\
&  \; \; \; - 2  \Im \SCP{\Psi_{N,t}}{\left( N^{-1} \vA^2(x_1) - \vAc^2(x_1,t) \right) q^{\varphi_t}_1 \Psi_{N,t}} 
\nonumber \\
&  \; \; \;  - 2 \Im  \SCP{\Psi_{N,t}}{ \left( (N-1)/N v(x_1 - x_2)  -  \left(v * \abs{\varphi_t}^2 \right)(x_1)  \right) q^{\varphi_t}_1 \Psi_{N,t}} .
\end{align}
Inserting the identity $1 = p_1 + q_1$ and the relations
\begin{itemize}
\item[]
$\Re \SCP{\Psi_N}{q_1 \left( N^{-1/2} \vA(x_1) - \vAc(x_1,t) \right) \cdot \nabla_1 q_1 \Psi_N} = 0$,
\item[]
$\Im \SCP{\Psi_N}{q_1 \left(  N^{-1} \vA^2(x_1) - \vAc^2(x_1,t) \right) q_1 \Psi_N} = 0$,
\item[]
$\Im  \SCP{\Psi_N}{q_1 \left(  (N-1) N^{-1} v(x_1 - x_2)  -  \left(v * \abs{\varphi}^2 \right)(x_1)  \right) q_1  \Psi_N} = 0$,
\end{itemize}
lead to
\begin{align}
\label{eq: Pauli alpha-a-nabla}
d_t \beta^a(t)  =&  - 4  \Re \SCP{\Psi_N}{p_1 \left(  N^{-1/2} \vA(x_1) - \vAc(x_1,t) \right) \cdot \nabla_1 q_1 \Psi_N}  \\
\label{eq: Pauli alpha-a-field}
&- 2 \Im \SCP{\Psi_N}{p_1 \left( N^{-1} \vA^2(x_1) - \vAc^2(x_1,t) \right) q_1\Psi_N}  \\
\label{eq: Pauli alpha-a-direct interaction}
&- 2  \Im  \SCP{\Psi_N}{ p_1 \left( (N-1) N^{-1} v(x_1-x_2)  -  \left(v * \abs{\varphi_t}^2 \right)(x_1)  \right)  q_1 \Psi_N} .
\end{align}
In the following, we estimate each line separately. To simplify the presentation we use the shorthand 
notation~\eqref{eq: Pauli shorthand notation field differences}.

\subsubsection{Bound on \eqref{eq: Pauli alpha-a-nabla}:}
Integration by parts and triangular inequality let us estimate
\begin{align}
\abs{\eqref{eq: Pauli alpha-a-nabla}} 
&\leq 4 \abs{\SCP{\Psi_N}{p_1\left(  \vAdp(x_1,t) 
+ \vAdm(x_1,t)   \right) \cdot \nabla_1 q_1 \Psi_N}}  
\nonumber \\
\label{eq: Pauli alpha-a-nabla-1}
&\leq  4 \abs{\SCP{\nabla_1 p_1 \Psi_N}{ \vAdm(x_1,t)  q_1 \Psi_N}}  \\
\label{eq: Pauli alpha-a-nabla-2}
&+ 4 \abs{\SCP{\nabla_1 p_1 \Psi_N}{ \vAdp(x_1,t)  q_1 \Psi_N}}.
\end{align}
By means of Lemma~\ref{lemma: Pauli auxiliary fields}, we bound the first line by
\begin{align}
\eqref{eq: Pauli alpha-a-nabla-1} 
&= 4 \abs{\SCP{\nabla_1 p_1 \Psi_N}{ \vAdm(x_1,t)   q_1 \Psi_N}} 
= 4 \abs{\SCP{\nabla_1 p_1 \Psi_N}{ \left( \eta * \vEdm \right) (x_1,t)  q_1 \Psi_N}} 
\nonumber \\
&= 4 \abs{\SCP{\vEdp(y,t) \nabla_1 p_1 \Psi_N}{ \eta( x_1 - y)   q_1 \Psi_N}_{;y}} 
\nonumber \\
&\leq  4 \norm{\vEdp(y,t) \cdot \nabla_1 p_1 \Psi_N}_{;y}  \norm{\eta(y-x_1)   q_1 \Psi_N}_{;y} 
\nonumber \\
&\leq 2 \norm{\vEdp(y,t) \cdot \nabla_1 p_1 \Psi_N}^2_{;y} 
+ 2  \norm{\eta}_2^2 \norm{q_1 \Psi_N}^2  
\nonumber \\
&\leq  \Lambda \pi^{-2} \beta^a +   C(\norm{\nabla \varphi}_2)  \beta^b
\leq  \Lambda   C(\norm{\nabla \varphi}_2)   \beta,
\end{align}
where we made use of Lemma~\ref{lemma: Pauli field operator times p or q against functional} and \eqref{eq: Pauli norm cut off function 2}. \\
The second term is bounded by
\begin{align}
\eqref{eq: Pauli alpha-a-nabla-2} 
&=  4 \abs{\SCP{\nabla_1 p_1 \Psi_N}{\int d^3y \, \eta(x_1 - y) \vEdp(y,t)  q_1 \Psi_N}}  
\nonumber \\
&=  4   \abs{\SCP{\nabla_1 p_1 \Psi_N}{ \eta(x_1 - y) \vEdp(y,t)  q_1 \Psi_N}_{;y}} 
\nonumber \\
&=  4  \abs{\SCP{q_1 \eta(x_1 -y) \nabla_1 p_1 \Psi_N}{ \vEdp(y,t)  \Psi_N}_{;y}} 
\nonumber \\
&= 4 \abs{\SCP{N^{-1} \sum_{i=1}^N q_i \eta(x_i - y) \nabla_i p_i \Psi_N}{\vEdp(y,t)  \Psi_N}_{;y}} 
\nonumber \\
&\leq 2 \norm{\vEdp(y,t)  \Psi_N}_{;y}^2
+2 \Big| \Big|  N^{-1} \sum_{i=1}^N q_i \eta(x_i - y) \nabla_i p_i \Psi_N \Big| \Big|_{;y}^2.
\end{align}
Lemma~\ref{lemma: Pauli field operator times p or q against functional} and the symmetry of the wave function lead to
\begin{align}
\eqref{eq: Pauli alpha-a-nabla-2} 
&\leq 2 \beta^b 
+ 2 N^{-2}   \SCP{\sum_{i=1}^N q_i \eta(x_i -y) \nabla_i p_i \Psi_N}{\sum_{j=1}^N q_j \eta(x_j -y) \nabla_j p_j \Psi_N}_{;y} 
\nonumber \\
&\leq  2 \beta^b 
+ 2 N^{-1} \norm{q_1 \eta(x_1-y) \nabla_1 p_1 \Psi_N}_{;y}^2
\nonumber\\
&+ 2 \SCP{q_1 \eta(x_1-y)\nabla_1 p_1 \Psi_N}{q_2 \eta(x_2-y) \nabla_2 p_2 \Psi_N}_{;y}
\nonumber\\
&\leq 2 \beta^b 
+ 2 N^{-1} \norm{ \eta(x_1-y) \nabla_1 p_1 \Psi_N}_{;y}^2
\nonumber\\
&+ 2 \SCP{\eta(x_1-y) \nabla_1 p_1 q_2 \Psi_N}{\eta(x_2-y) \nabla_2 p_2 q_1 \Psi_N}_{;y}
\nonumber \\
&\leq  2 \beta^b 
+ 2 N^{-1} \norm{ \eta(x_1-y) \nabla_1 p_1 \Psi_N}_{;y}^2
\nonumber\\
&+ 2 \norm{\eta(x_1-y) \nabla_1 p_1 q_2 \Psi_N}_{;y} \norm{\eta(x_2-y) \nabla_2 p_2 q_1 \Psi_N}_{;y}
\nonumber\\
&\leq  2 \beta^b 
+2 N^{-1} \SCP{\eta(x_1-y) \nabla_1 p_1 \Psi_N}{\eta(x_1-y) \nabla_1 p_1 \Psi_N}_{;y}
\nonumber \\
&+2 \SCP{\eta(x_1-y) \nabla_1 p_1 q_2 \Psi_N}{\eta(x_1-y) \nabla_1 p_1 q_2 \Psi_N}_{;y}.
\end{align}
Interchanging the order of integration we have
\begin{align}
\eqref{eq: Pauli alpha-a-nabla-2} 
&\leq  2 N^{-1} \SCP{\nabla_1 p_1 \Psi_N}{\left( \int d^3y \, \abs{\eta(x_1-y)}^2  \right) \nabla_1 p_1 \Psi_N}
+ 2 \beta^b 
\nonumber\\
& \; \; \; \; +  2 \SCP{\nabla_1 p_1 q_2 \Psi_N}{\left( \int d^3y \, \abs{\eta(x_1 -y)}^2  \right) \nabla_1 p_1 q_2 \Psi_N}
\nonumber\\
&=  2  \norm{\eta}_2^2 \left(  N^{-1} \SCP{\Psi_N}{p_1 (- \Delta_1) p_1 \Psi_N} + \SCP{\Psi_N}{q_2 p_1 (- \Delta_1) p_1 q_2 \Psi_N}  \right)  + 2 \beta^b .
\end{align}
By virtue of $p_1 \left( - \Delta \right) p_1 = p_1 \norm{\nabla \varphi}_2^2$, this becomes
\begin{align}
\eqref{eq: Pauli alpha-a-nabla-2}
&\leq
 2  \norm{\eta}_2^2  \norm{\nabla \varphi}_2^2 \left( N^{-1} \SCP{\Psi_N}{p_1 \Psi_N}   
 +  \SCP{ \Psi_N}{ q_2 p_1q_2 \Psi_N} \right)  + 2 \beta^b 
\nonumber \\
&\leq   \norm{\eta}_2^2  C(\norm{\varphi}_{H^2})  \left( \beta^a 
+ \beta^b + N^{-1} \right) 
\leq  \Lambda C(\norm{\varphi}_{H^2}) \left( \beta + N^{-1} \right) 
\end{align}
and we obtain
\begin{align}
\abs{\eqref{eq: Pauli alpha-a-nabla}} 
&\leq  \Lambda  C(\norm{\varphi}_{H^2})  \left( \beta + N^{-1} \right).
\end{align}

\subsubsection{Bound on \eqref{eq: Pauli alpha-a-field}:}

\begin{align}
\abs{\eqref{eq: Pauli alpha-a-field}}
&\leq 2  \abs{\SCP{\Psi_N}{p_1 \left( N^{-1} \vA^2(x_1) - \vAc^2(x_1,t) \right)q_1 \Psi_N}}  
\nonumber \\
&= 2  \abs{\SCP{\Psi_N}{p_1 \left( N^{-1/2} \vA(x_1) - \vAc(x_1,t) \right) \left( N^{-1/2} \vA(x_1) + \vAc(x_1,t) \right)q_1 \Psi_N}}  
\nonumber \\
\label{eq: Pauli alpha-a-field-1}
&\leq 2  \abs{\SCP{\Psi_N}{p_1 \vAdm(x_1,t) \left(  N^{-1/2} \vA(x_1) + \vAc(x_1,t) \right)q_1 \Psi_N}}  \\
\label{eq: Pauli alpha-a-field-2}
&+ 2  \abs{\SCP{\Psi_N}{p_1 \vAdp(x_1,t) \left( N^{-1/2} \vA(x_1) + \vAc(x_1,t) \right)q_1 \Psi_N}}  
\end{align}
First, we deal with line \eqref{eq: Pauli alpha-a-field-1}:
\begin{align}
\eqref{eq: Pauli alpha-a-field-1} 
&= 2 \abs{\SCP{\Psi_N}{p_1 \left( \eta * \vEdm\right)(x_1,t) \left( N^{-1/2} \vA(x_1) + \vAc(x_1,t) \right)q_1 \Psi_N}} 
\nonumber \\
&= 2  \abs{\SCP{\vEdp(y,t) p_1 \Psi_N}{ \eta(y-x_1)\left( N^{-1/2} \vA(x_1) + \vAc(x_1,t) \right)q_1 \Psi_N}_{;y}}  \nonumber \\
&\leq  \norm{\vEdp(y,t) p_1 \Psi_N}_{;y}^2 
+   \norm{\eta(x_1-y) \left( N^{-1/2} \vA(x_1) + \vAc(x_1,t) \right)q_1 \Psi_N}_{;y}^2 
\nonumber \\
&\leq \norm{\vEdp(y,t) p_1 \Psi_N}_{;y}^2 
+ \norm{\eta}_2^2 \norm{\left( N^{-1/2} \vA(x_1) + \vAc(x_1,t) \right)q_1 \Psi_N}^2.
\end{align}
Making use of Lemma~\ref{lemma: Pauli field operator times p or q against functional} and $(a+b)^2 \leq 2 (a^2 + b^2)$, we obtain
\begin{align}
\eqref{eq: Pauli alpha-a-field-1}
&\leq \beta^b 
+ 2 \norm{\eta}_2^2 \left(  \norm{\vAc}_{\infty}^2 \norm{q_1 \Psi_N}^2 + \norm{ N^{-1/2} \vA(x_1) q_1 \Psi_N}^2 \right). 
\end{align}
By means of \eqref{eq: Pauli norm cut off function 2} and Lemma~\ref{lemma: Pauli Feldoperator gegen alpha Abschaetzung}, this becomes
\begin{align}
\eqref{eq: Pauli alpha-a-field-1}
\leq&  \Lambda^2 C(\norm{\vAc}_{\infty} , \norm{\varphi}_{H^2}, \norm{\varphi}_{\infty}, \mathcal{E}_M)  \left(  \beta + \Lambda/N \right).
\end{align}
The second line is bounded by
\begin{align}
\eqref{eq: Pauli alpha-a-field-2}
&=  2  \abs{\SCP{\Psi_N}{p_1 \vAdp(x_1,t) \left( N^{-1/2} \vA(x_1) + \vAc(x_1,t) \right)q_1 \Psi_N}}  
\nonumber \\
&=  2  \abs{\SCP{\Psi_N}{p_1\left[ \left( N^{-1/2} \vA(x_1) + \vAc(x_1,t) \right) \vAdp(x_1,t) + \Lambda^2/(4 \pi^2 N)   \right] q_1 \Psi_N}}  
\nonumber \\
&\leq 2  \abs{\SCP{\Psi_N}{p_1\left( N^{-1/2} \vA(x_1) + \vAc(x_1,t) \right) \vAdp(x_1,t)   q_1 \Psi_N}} 
\nonumber \\
&+ 2 \Lambda^2/(4 \pi^2 N) \abs{\underbrace{\SCP{\Psi_N}{p_1q_1 \Psi_N}}_{=0}} .
\end{align}
Here, we have we used the commutation relation
\begin{align}
&\left[ \vAdp(x_1,t) , \left( N^{-1/2} \vA(x_1) + \vAc(x_1,t) \right)  \right]  
= N^{-1} \left[ \vAp(x_1) , \vAm(x_1)  \right]  
= \Lambda^2/(4 \pi^2 N).
\end{align}
Lemma~\ref{lemma: Pauli auxiliary fields} and  Lemma~\ref{lemma: Pauli field operator times p or q against functional} lead to
\begin{align}
\eqref{eq: Pauli alpha-a-field-2} 
&\leq  2  \abs{\SCP{\left( N^{-1/2} \vA(x_1) + \vAc(x_1,t) \right) p_1 \Psi_N}{ \int d^3y \, \eta(x_1-y) \vEdp(y,t)   q_1 \Psi_N}} 
\nonumber\\
&= 2 \abs{\SCP{q_1 \eta(x_1 - y) \left( N^{-1/2} \vA(x_1) + \vAc(x_1,t) \right) p_1 \Psi_N}{\vEdp(y,t) \Psi_N}_{;y}}
\nonumber\\
&= 2 N^{-1} \abs{\SCP{\sum_{i=1}^N q_i \eta(x_i - y) \left( N^{-1/2} \vA(x_i) + \vAc(x_i,t) \right) p_i \Psi_N}{\vEdp(y,t) \Psi_N}_{;y}}
\nonumber\\
&\leq N^{-2} \Big| \Big| \sum_{i=1}^N q_i \eta(x_i - y) \left( N^{-1/2} \vA(x_i) + \vAc(x_i,t) \right) p_i \Psi_N \Big| \Big|_{;y}^2
+ \norm{\vEdp(y,t) \Psi_N}_{;y}^2
\nonumber\\
&\leq N^{-2} \Big| \Big|\sum_{i=1}^N q_i \eta(x_i - y) \left( N^{-1/2} \vA(x_i) + \vAc(x_i,t) \right) p_i \Psi_N \Big| \Big|_{;y}^2
+ \beta^b
\end{align}
Similar to the estimate of~\eqref{eq: Pauli alpha-a-nabla-2} one obtains
\begin{align}
\eqref{eq: Pauli alpha-a-field-2} 
&\leq  N^{-1}  \norm{\eta(x_1-y) \left( N^{-1/2} \vA(x_1) + \vAc(x_1,t) \right) p_1 \Psi_N}_{;y}^2
\nonumber\\
& \qquad \,  +  \norm{\eta(x_1-y) \left( N^{-1/2} \vA(x_1) + \vAc(x_1,t) \right) p_1 q_2 \Psi_N}_{;y}^2   + \beta^b
\nonumber\\
&=
N^{-1} \norm{\eta}_2^2  \norm{\left( N^{-1/2} \vA(x_1) + \vAc(x_1,t) \right) p_1 \Psi_N}^2
\nonumber\\
& \qquad \,  + \norm{\eta}_2^2  \norm{\left( N^{-1/2} \vA(x_1) + \vAc(x_1,t) \right) p_1 q_2 \Psi_N}^2   + \beta^b
\nonumber\\
&\leq C \Lambda  \left(  \beta^b + \norm{\vAc}_{\infty}^2 \beta^a + \norm{N^{-1/2} \vAc(x_1) p_1 q_2 \Psi_N}^2  \right)   
\nonumber\\
&+ C \Lambda/N  \left( \norm{\vAc}_{\infty}^2 +  \norm{N^{-1/2} \vAc(x_1) p_1 \Psi_N}^2  \right) .
\end{align}

By means of Lemma~\ref{lemma: Pauli Feldoperator gegen alpha Abschaetzung} this is bounded by
\begin{align}
\eqref{eq: Pauli alpha-a-field-2} 
&\leq \Lambda^2 C(\norm{\vAc}_{\infty}, \norm{\varphi}_{H^2}, \norm{\varphi}_{\infty}, \mathcal{E}_M)  
\left( \beta + \Lambda/N \right).
\end{align}

In total, we obtain
\begin{align}
\abs{\eqref{eq: Pauli alpha-a-field}}
&\leq \eqref{eq: Pauli alpha-a-field-1} + \eqref{eq: Pauli alpha-a-field-2} 
\leq  \Lambda^2  C(\norm{\vAc}_{\infty}, \norm{\varphi}_{H^2}, \norm{\varphi}_{\infty}, \mathcal{E}_M)      \left( \beta +  \Lambda/N  \right).
\end{align}
}

\subsubsection{Bound on \eqref{eq: Pauli alpha-a-direct interaction}:}
Subsequently, we consider the term that arises from the direct interaction. Inserting the identity $1 = p_2 + q_2$ and using the shorthand shorthand notation 
\begin{align}
Z(x_1,x_2) := (N-1) N^{-1} v(x_1-x_2) - \left( v * \abs{\varphi}^2 \right)(x_1)
\end{align}
gives
\begin{align}
\eqref{eq: Pauli alpha-a-direct interaction}
=& -2 \Im \SCP{\Psi_N}{p_1 Z(x_1,x_2) q_1 \Psi_N} 
\nonumber \\
=& -2 \Im \SCP{\Psi_N}{p_1 p_2 Z(x_1,x_2) q_1 p_2 \Psi_N} 
\nonumber \\
& -2 \Im \SCP{\Psi_N}{p_1 p_2 Z(x_1,x_2) q_1 q_2 \Psi_N} 
\nonumber \\
& -2 \Im \SCP{\Psi_N}{p_1 q_2 Z(x_1,x_2) q_1 p_2 \Psi_N} 
\nonumber \\
& -2 \Im \SCP{\Psi_N}{p_1 q_2 Z(x_1,x_2) q_1 q_2 \Psi_N}  .
\end{align}
The third term vanishes due to symmetry of the wave function under the interchange of $x_1$ and $x_2$ and we are left with
\begin{align}
\label{eq: Pauli alpha-a-direct interaction-1}
\abs{\eqref{eq: Pauli alpha-a-direct interaction}}
\leq& 2 \abs{\SCP{\Psi_N}{p_1 p_2 Z(x_1,x_2) q_1 p_2 \Psi_N}}  \\
\label{eq: Pauli alpha-a-direct interaction-2}
+& 2 \abs{\SCP{\Psi_N}{p_1 p_2 Z(x_1,x_2) q_1 q_2 \Psi_N}}      \\
\label{eq: Pauli alpha-a-direct interaction-3}
+& 2 \abs{\SCP{\Psi_N}{p_1 q_2 Z(x_1,x_2) q_1 q_2 \Psi_N}}   .
\end{align}
The first line is the most important. It is small because the direct interaction of the many-body system is well approximated by the mean-field potential. By means of 
\begin{align}
p_2 Z(x_1,x_2) p_2 
&= p_2 \left[  (N-1)N^{-1} v(x_1-x_2) - \left( v * \abs{\varphi}^2 \right)(x_1) \right] p_2 
\nonumber \\
&= \; \; \, \, \, \left[ (N-1)N^{-1}  -1 \right]   \left( v * \abs{\varphi}^2 \right)(x_1) p_2 
= - N^{-1} \left( v * \abs{\varphi}^2 \right)(x_1)  p_2 
\end{align}
one has
\begin{align}
\abs{\eqref{eq: Pauli alpha-a-direct interaction-1}}
\leq& 2 N^{-1} \abs{\SCP{\Psi_N}{p_1 \left( v * \abs{\varphi}^2 \right)(x_1) p_2 q_1 \Psi_N}} 
\nonumber \\
\leq& 2 N^{-1} \norm{\left( v * \abs{\varphi}^2 \right)(x_1) p_1 \Psi_N}
\norm{p_2 q_1 \Psi_N} 
\leq 2 N^{-1} \norm{v * \abs{\varphi}^2}_{\infty}.
\end{align}
We decompose the interaction potential $v = v_1 + v_2$ into $v_1 \in L^2(\mathbb{R}^3)$ and $v_2 \in L^{\infty}(\mathbb{R}^3)$. Then 
\begin{align}
\label{eq: Pauli direct interaction potential Young 2}
\norm{v * \abs{\varphi}^2}_{\infty}
&\leq  \norm{v_1 * \abs{\varphi}^2}_{\infty} + \norm{v_2 * \abs{\varphi}^2}_{\infty}
\leq \norm{v_1}_2 \norm{\abs{\varphi}^2}_2 + \norm{v_2}_{\infty} \norm{\abs{\varphi}^2}_1
\nonumber\\
&\leq \norm{v_1}_2 \norm{\varphi}_{\infty} \norm{\varphi}_2 + \norm{v_2}_{\infty} \norm{\varphi}_2^2
\leq C(\norm{\varphi}_{\infty}).
\end{align}
holds due to Young's inequality and we obtain
\begin{align}
\abs{\eqref{eq: Pauli alpha-a-direct interaction-1}} \leq N^{-1} C(\norm{\varphi}_{\infty}).
\end{align}
Moreover, we have
\begin{align}
p_1 Z^2(x_1,x_2) p_1
=& p_1 \scp{\varphi_t}{\left((N-1) N^{-1} v(x_2- \cdot) -
\left( v * \abs{\varphi_t}^2 \right) \right)^2 \varphi_t} 
\nonumber \\
\leq& 2 p_1 \scp{\varphi_t}{\left( v^2(x_2 - \cdot) 
+ \left( v * \abs{\varphi_t}^2  \right)^2 \right) \varphi_t}  
\nonumber \\
\leq& 2 p_1 \left( \norm{v^2 * \abs{\varphi_t}^2}_{\infty}
+ \norm{v * \abs{\varphi_t}^2}_{\infty}^2 \right) \leq  p_1 C(\norm{\varphi}_{\infty})
\end{align}
because of~\eqref{eq: Pauli direct interaction potential Young 1} and~\eqref{eq: Pauli direct interaction potential Young 2}.
This shows
\begin{align}
\norm{p_1 Z^2(x_1,x_2) p_1}_{\op} \leq C(\norm{\varphi}_{\infty})
\end{align}
and allows us to estimate
\begin{align}
\abs{\eqref{eq: Pauli alpha-a-direct interaction-2}}
&=
2 \abs{\SCP{q_2 Z(x_1,x_2) p_1 p_2 \Psi_N}{q_1 \Psi_N}}
= 2 (N-1)^{-1} \abs{\SCP{\sum_{i=2}^N q_i Z(x_1,x_i) p_1 p_i \Psi_N}{q_1 \Psi_N}}
\nonumber \\
&\leq N^{-1} \Big| \Big| \sum_{i=2}^N q_i  Z(x_1,x_i) p_1 p_i \Psi_N \Big| \Big|^2 + 4 \norm{q_1 \Psi_N}^2
\nonumber\\
&= N^{-2} \SCP{\sum_{i=2}^N q_i Z(x_1,x_i) p_1 p_i \Psi_N}{\sum_{j=2}^N q_j Z(x_1,x_j) p_1 p_j \Psi_N}
+ 4 \beta^a
\nonumber\\
&\leq \SCP{q_2 Z(x_1,x_2) p_1 p_2 \Psi_N}{q_3 Z(x_1,x_3) p_1 p_3 \Psi_N}
+ N^{-1} \norm{q_2 Z(x_1,x_2) p_1 p_2 \Psi_N}^2 + 4 \beta^a
\nonumber\\
&\leq \SCP{Z(x_1,x_2)p_1 p_2 q_3 \Psi_N}{Z(x_1,x_3)p_1 p_3 q_2 \Psi_N}
+ N^{-1} \norm{ Z(x_1,x_2) p_1 p_2 \Psi_N}^2 + 4 \beta^a
\nonumber\\
&\leq \norm{Z(x_1,x_2)p_1 p_2 q_3 \Psi_N}^2
+ N^{-1} \norm{ Z(x_1,x_2) p_1 p_2 \Psi_N}^2 + 4 \beta^a
\nonumber\\
&\leq  \norm{p_1 Z^2(x_1,x_2) p_1}_{op} \left( \beta^a + N^{-1} \right) + 4 \beta^a
\nonumber\\
&\leq C(\norm{\varphi}_{\infty}) \left( \beta + N^{-1} \right).
\end{align}
The last term of \eqref{eq: Pauli alpha-a-direct interaction} is bounded by
\begin{align}
\abs{\eqref{eq: Pauli alpha-a-direct interaction-3}}
=& 2 \SCP{Z(x_1,x_2) p_1 q_2 \Psi_N}{q_1 q_2 \Psi_N}  
\nonumber \\
\leq& \SCP{\Psi_N}{q_2 p_1 Z^2(x_1,x_2) p_1 q_2 \Psi_N} + \norm{q_1 q_2 \Psi_N}^2  
\nonumber \\
\leq& \norm{p_1 Z^2(x_1,x_2) p_1}_{\op} \norm{q_2 \Psi_N}^2 + \beta^a
\leq C(\norm{\varphi}_{\infty}) \beta.
\end{align}
This leads to
\begin{align}
\abs{\eqref{eq: Pauli alpha-a-direct interaction}}
\leq C(\norm{\varphi}_{\infty}) \left( \beta + N^{-1} \right).
\end{align}

\subsection{Bound on $d_t \beta^b$:}

\begin{lemma}
\label{lemma: dt beta-b}
Let $v$ satisfy (A1), $(\varphi_t,\boldsymbol{A}(t), \boldsymbol{E}(t)) \in \mathcal{G}$ and  $\Psi_{N,t}$ be the unique solution of \eqref{eq: Pauli Schroedinger equation microscopic} with initial data $\Psi_{N,0} \in \left(L^2_s \left( \mathbb{R}^{3N} \right) \otimes \mathcal{F}_p \right) \cap \mathcal{D}(H_N)$. 
 Then, there exists a monotone increasing function $C(t)$ of $\mathcal{E}_M[\varphi_t,u_t]$, $\norm{\varphi_t}_{H^2(\mathbb{R}^3)}$ and $\norm{\varphi_t}_{L^{\infty}(\mathbb{R}^3)}$  such that
\begin{align}
\abs{d_t \beta^b(\Psi_{N,t}, u_t)} \leq  \Lambda^4  C(t)  \left( \beta(\Psi_{N,t},\varphi_t, u_t) + \Lambda/N  \right).
\end{align}
\end{lemma}

\proof{
We would like to note that the following calculation can be carried out in more detail. We could for example write $\beta^b$ as
\begin{align}
\beta^b(\Psi_{N,t}, u_t) &= N^{-1} \SCP{\Psi_{N,t}}{H_f \Psi_{N,t}} 
+ \norm{u_t}_{\mathfrak{h}}^2
\nonumber \\
&- 2 N^{-1/2} \Re \SCP{\Psi_{N,t}}{ \big(\sum_{j=1,2} \int d^3k \,  u_t(k,\lambda) \abs{k}^{1/2} a^*(k,\lambda)  \big) \Psi_{N,t}}
\end{align}
and determine its derivative analogously to \cite[Appendix 2.11]{leopold2} . Since we disregard photons with small energies, $\beta^b$ is well defined for $(\varphi_t, \boldsymbol{A}(t),  \boldsymbol{E}(t)) \in \mathcal{G}$ and $\Psi_{N,t} \in  \mathcal{D}(H_N) = \mathcal{D} ( \sum_{i=1}^N (- \Delta_i) + H_f) \subset \mathcal{D}(H_f)$. This allows us to determine the derivative for many-body wave functions in $\mathcal{D}(H_N^2)$ (which is invariant due to Stone's theorem) and extend the result later to $\mathcal{D}(H_N)$ by a standard density argument. \\
We compute the commutator
\begin{align}
i \left[ H_N,    \frac{a(k,\lambda)}{\sqrt{N}} \right]
&= - i \abs{k}  \frac{a(k,\lambda)}{\sqrt{N}} 
- \frac{2i}{N} \sum_{j=1}^N \frac{\tilde{\kappa}(k)}{\sqrt{2 \abs{k}}}  \vep_{\lambda}(k) e^{-ikx_j} \left( i \nabla_j + \frac{\vA(x_j)}{\sqrt{N}} \right)
\end{align}
by use of the canonical commutation relations~\eqref{eq: canonical commutation relation} and observe that the Maxwell-Schr\"odinger system leads to
\begin{align}
\partial_t \abs{k}^{1/2} \alpha_t(k,\lambda)
= - i \abs{k}^{3/2} \alpha_t(k,\lambda) + \frac{i}{\sqrt{2}} \tilde{\kappa}(k) \vep_{\lambda}(k) (2 \pi)^{3/2} \mathcal{FT}[\vj](k).
\end{align}
Then, we continue with
\footnotesize
\begin{align}
d_t \beta^b(t) &= \sum_{\lambda=1,2} \int d^3k \, 
 \abs{k} d_t \SCP{\left( N^{-1/2} a(k,\lambda) - \alpha_t(k,\lambda) \right) \Psi_N}{\left( N^{-1/2} a(k,\lambda) - \alpha_t(k,\lambda) \right) \Psi_N}    
\nonumber \\
&=  \sum_{\lambda=1,2} \int d^3k \, 
 \abs{k} \SCP{ i \left[ H_N, N^{-1/2} a(k,\lambda) \right] \Psi_N}{\left( N^{-1/2} a(k,\lambda) - \alpha_t(k,\lambda) \right) \Psi_N}    
\nonumber \\
&+ \sum_{\lambda=1,2} \int d^3k \, 
d_t \abs{k} \SCP{\left( N^{-1/2} a(k,\lambda) - \alpha_t(k,\lambda) \right) \Psi_N}{i \left[ H_N, N^{-1/2} a(k,\lambda) \right] \Psi_N}    
\nonumber \\
&- \sum_{\lambda=1,2} \int d^3k \, 
\abs{k} \SCP{ (\partial_t \alpha_t)(k,\lambda)  \Psi_N}{\left( N^{-1/2} a(k,\lambda) - \alpha_t(k,\lambda) \right) \Psi_N}    
\nonumber \\
&- \sum_{\lambda=1,2} \int d^3k \, 
\abs{k} \SCP{\left( N^{-1/2} a(k,\lambda) - \alpha_t(k,\lambda) \right) \Psi_N}{(\partial_t \alpha_t)(k,\lambda) \Psi_N}    
\nonumber \\ 
&= 2  \sum_{\lambda=1,2} \int d^3k \, 
 \abs{k} \Re \SCP{ i \left[ H_N, N^{-1/2} a(k,\lambda) \right] \Psi_N}{\left( N^{-1/2} a(k,\lambda) - \alpha_t(k,\lambda) \right) \Psi_N}    
\nonumber \\
&- 2 \sum_{\lambda=1,2} \int d^3k \, 
\abs{k} \Re \SCP{ (\partial_t \alpha_t)(k,\lambda)  \Psi_N}{\left( N^{-1/2} a(k,\lambda) - \alpha_t(k,\lambda) \right) \Psi_N}    
\nonumber \\
&= 2 \sum_{\lambda=1,2} \int d^3k \, 
\abs{k}^2 \Re \SCP{ - i \left( N^{-1/2} a(k,\lambda) - \alpha_t(k,\lambda) \right) \Psi_N}{\left( N^{-1/2} a(k,\lambda) - \alpha_t(k,\lambda) \right) \Psi_N}    
\nonumber \\
&+ 4 \sum_{\lambda=1,2} \int d^3k \, 
 \abs{k} \Re \SCP{ - \frac{i}{N} \sum_{j=1}^N \frac{\tilde{\kappa}(k)}{\sqrt{2 \abs{k}}} \vep_{\lambda}(k) e^{-ikx_j} \left( i \nabla_j +    \frac{\vA(x_j)}{\sqrt{N}} \right) \Psi_N}{\left( N^{-1/2} a(k,\lambda) - \alpha_t(k,\lambda) \right) \Psi_N}    
\nonumber \\
&+ 2 \sum_{\lambda=1,2} \int d^3k \, 
 \abs{k} \SCP{ -i \frac{\tilde{\kappa}(k)}{\sqrt{2 \abs{k}}}  \vep_{\lambda}(k) (2 \pi)^{3/2} \mathcal{FT}[j](k) \Psi_N}{\left( N^{-1/2} a(k,\lambda) - \alpha_t(k,\lambda) \right) \Psi_N} .   
\end{align}
\normalsize
The first terms cancels because the scalar product is purely imaginary. \\
\noindent
So if we use $\left[ \nabla_1, \vep_{\lambda}(k) e^{ikx_1} \right] =0$ (recall Definition~\eqref{eq: polarization vectors}) and the symmetry of the wave function, we get
\footnotesize
\begin{align}
d_t \beta^b(t) &= 
 4 \Re \sum_{\lambda=1,2} \int d^3k \,
\SCP{\Psi_N}{ i \sqrt{\frac{\abs{k}}{2}}\tilde{\kappa}(k) \vep_{\lambda}(k) e^{ikx_1} \left( i \nabla_1 +  N^{-1/2} \vA(x_1) \right) \left(  N^{-1/2} a(k,\lambda) - \alpha_t(k,\lambda) \right) \Psi_N} 
\nonumber \\
&+ 2 \Re \sum_{\lambda=1,2} \int d^3k \,
\SCP{\Psi_N}{i \sqrt{\frac{\abs{k}}{2}} \tilde{\kappa}(k) \vep_{\lambda}(k) (2 \pi)^{3/2} \mathcal{FT}[\vj]^*(k) \left(   N^{-1/2} a(k,\lambda) - \alpha_t(k,\lambda) \right) \Psi_N}.
\end{align}
\normalsize
Inserting the identity $1= p_1 + q_1$ and
\begin{align}
\sum_{\lambda=1,2} \int d^3k  \,  i  \sqrt{ \frac{\abs{k}}{2}}\tilde{\kappa}(k) \vep_{\lambda}(k) e^{ikx_1}
\left(  N^{-1/2} a(k,\lambda) - \alpha_t(k,\lambda) \right)
=  \left( \kappa * \vEdp \right)(x_1,t) 
\end{align}
lead to
\begin{align}
d_t \beta^b(t) =
&+ 4 \Re \SCP{\Psi_N}{p_1 N^{-1/2} \vA(x_1) \kappa(y-x_1) p_1 \vEdp(y,t) \Psi_N}_{;y} 
 \nonumber \\
&+ 2 \Re \SCP{\Psi_N}{p_1 \left( \kappa(y-x_1) i \nabla_1 + i \nabla_1 \kappa(y-x_1) \right) p_1 \vEdp(y,t) \Psi_N}_{;y} 
\nonumber \\
&+ 2 \Re \SCP{\Psi_N}{\Big( \int d^3z \, \kappa(y-z) \vj(z) \Big) \vEdp(y,t) \Psi_N}_{;y}  
\nonumber \\
&+ 4 \Re \SCP{\Psi_N}{q_1 \kappa(y-x_1)  i \nabla_1  p_1  \vEdp(y,t) \Psi_N}_{;y}   
\nonumber \\
&+ 4 \Re \SCP{\Psi_N}{q_1 N^{-1/2} \vA(x_1) \kappa(y-x_1)  p_1\vEdp(y,t) \Psi_N}_{;y}   
\nonumber \\
&+ 4 \Re \SCP{\Psi_N}{ \left( i \nabla_1 +  N^{1/2} \vA(x_1) \right) \kappa(y-x_1)  q_1  \vEdp(y,t) \Psi_N}_{;y}  .  
\end{align}
With the relations 
\begin{align}
p_1 N^{-1/2}  \vA(x_1) \kappa(y-x_1) p_1 
&= p_1 \int d^3z \, \abs{\varphi}^2(z) N^{-1/2} \vA(z) \kappa(y-z)  ,
\nonumber \\
p_1 \left( \kappa(y-x_1) i \nabla_1 + i \nabla_1 \kappa(y-x_1)  \right) p_1
&= -2 p_1 \int d^3z \,  \kappa(y-z) \Im[\varphi^* 
\nabla \varphi](z) ,
\nonumber \\
\vj &= 2 \left(  \Im(\varphi^* \nabla \varphi) - \abs{\varphi}^2 \vAc \right)
\end{align}
we obtain
\begin{align}
\label{eq: beta-b-5}
d_t \beta^b(t) =
&- 4 \Re \int d^3z \, \abs{\varphi}^2(z) \SCP{\Psi_N}{q_1 N^{-1/2} \vA(z) \kappa(y-z) 
\vEdp(y,t) \Psi_N}_{;y}  \\
\label{eq: beta-b-6}
&+ 4 \Re \int d^3z \, \Im[\varphi^* \nabla \varphi](z)
\SCP{\Psi_N}{q_1 \kappa(y-z) \vEdp(y,t) \Psi_N}_{;y}  \\
\label{eq: beta-b-7}
&+ 4 \Re \int d^3z \, \abs{\varphi}^2(z)
\SCP{\Psi_N}{\kappa(y-z) \left(  N^{-1/2} \vA(z) - \vAc(z,t) \right) \vEdp(y,t) \Psi_N}_{;y}  \\
\label{eq: beta-b-8}
&+ 4 \Re \SCP{\Psi_N}{q_1 \kappa(y-x_1) i \nabla_1 p_1 \vEdp(y,t) \Psi_N}_{;y}  \\
\label{eq: beta-b-9}
&+ 4 \Re \SCP{\Psi_N}{q_1 N^{-1/2} \vA(x_1)  \kappa(y-x_1) p_1 \vEdp(y,t) \Psi_N}_{;y}  \\
\label{eq: beta-b-10}
&+ 4 \Re \SCP{\Psi_N}{ \left( -i \nabla_1 -  N^{-1/2} \vA(x_1) \right) \kappa(y-x_1) q_1 \vEdp(y,t) \Psi_N}_{;y}  .
\end{align}
Subsequently, we estimate each line separately:
\begin{align}
\abs{\eqref{eq: beta-b-5}}
&\leq  4 \abs{\int d^3z \,  \abs{\varphi}^2(z) \SCP{\Psi_N}{q_1 N^{-1/2} \vA(z) \kappa(y-z)  \vEdp(y,t) \Psi_N}_{;y}} 
\nonumber \\
&\leq 4  \int d^3y \int d^3z \, \abs{\varphi}^2(z) \abs{\kappa(y-z)}
\abs{\SCP{ N^{-1/2} \vA(z) q_1 \Psi_N}{ \vEdp(y,t) \Psi_N} }  
\nonumber  \\
&\leq 4 \int d^3y \int d^3z \, \abs{\varphi}^2(z)  \norm{\vEdp(y,t) \Psi_N} \abs{\kappa(y-z)} \norm{ N^{-1/2} \vA(z) q_1 \Psi_N} 
\nonumber \\
&\leq 2 \int d^3y \int d^3z \,  \abs{\varphi}^2(z)  \norm{\vEdp(y,t) \Psi_N}^2  
\nonumber \\
&+ 2 \int d^3z \, \abs{\varphi}^2(z) \norm{ N^{-1/2} \vA(z) q_1 \Psi_N}^2 
\Big( \int d^3y \, \abs{\kappa(y-z)}^2 \Big) 
\nonumber \\
&= 2 \SCP{\Psi_N}{\vEdm(y,t)  \vEdp(y,t) \Psi_N}_{;y}  
\nonumber \\
&+ 2 \norm{\kappa}_2^2  \int d^3z \,  \abs{\varphi}^2(z) \norm{ N^{-1/2} \vA(z) q_1 \Psi_N}^2.
\end{align}
With the help of Lemma~\ref{lemma: Pauli Feldoperator gegen alpha Abschaetzung}  and \eqref{eq: Pauli norm cut off function 2} this becomes
\begin{align}
\abs{\eqref{eq: beta-b-5}} 
&\leq  \Lambda^4  C(\norm{\varphi}_{H^2}, \norm{\varphi}_{\infty}, \mathcal{E}_M)  \left(  \beta + \Lambda/N \right).
\end{align}
Similarly,
\begin{align}
\abs{\eqref{eq: beta-b-6}} 
&\leq  4 \int d^3y \int d^3z \, \abs{\kappa(y-z)} \abs{\varphi(z)} \abs{\nabla \varphi(z)}
\abs{\SCP{q_1 \Psi_N}{\vEdp(y,t) \Psi_N}}  
\nonumber \\
&\leq 4  \int d^3y \int d^3z \, \abs{\kappa(y-z)} \abs{\varphi(z)} \abs{\nabla \varphi(z)}
\norm{\vEdp(y,t) \Psi_N} \norm{q_1 \Psi_N}   
\nonumber \\
&\leq 2  \int d^3y \int d^3z \, \abs{\nabla \varphi(z)}^2 
\norm{\vEdp(y,t) \Psi_N}^2  
\nonumber \\ 
&+ 2  \int d^3z \, \abs{\varphi(z)}^2 \norm{q_1 \Psi_N}^2 
\Big( \int d^3y \,  \abs{\kappa(y-z)}^2 \Big)    
\nonumber \\
&\leq 2 \norm{\nabla \varphi}_2^2 \beta^b + 2 \norm{\kappa}_2^2 \beta^a 
\leq  \Lambda^3  C(\norm{\varphi}_{H^2})  \beta  
\end{align}
and
\begin{align}
\abs{\eqref{eq: beta-b-7}}   \leq&
4 \int d^3y \int d^3z \, \abs{\varphi}^2(z) \abs{\kappa(y-z)}
\abs{\SCP{\vAd(z,t) \Psi_N}{\vEdp(y,t) \Psi_N}} 
\nonumber  \\
\leq&  4 \int d^3y \int d^3z \, \abs{\varphi}^2(z) \abs{\kappa(y-z)} \norm{\vAd(z,t) \Psi_N} \norm{\vEdp(y,t)}   
\nonumber \\
\leq& 2 \int d^3z \, \abs{\varphi}^2(z)  \norm{\vAd(z,t) \Psi_N}^2  \int d^3y \, \abs{\kappa(y-z)}^2 
\nonumber \\
+& 2 \int d^3z \, \abs{\varphi}^2(z) \int d^3y \, \norm{\vEdp(y,t)}^2 
\nonumber \\
\leq& 2 \beta^b + 2 \norm{\kappa}_2^2 
\int d^3z \abs{\varphi}^2(z)  \norm{\vAd(z,t) \Psi_N}^2.
\end{align}
Linearity and $(a+b)^2 \leq 2 (a^2 + b^2)$ lead to
\begin{align}
\abs{\eqref{eq: beta-b-7}}   
\leq&  2 \beta^b + 4 \norm{\kappa}_2^2 
\int d^3z \abs{\varphi}^2(z) \left( \norm{\vAdp(z,t) \Psi_N}^2  + \norm{\vAdm(z,t) \Psi_N}^2  \right).
\end{align}
By means of the commutation relation
\begin{align}
\left[ \vAdp(z,t) , \vAdm(z,t) \right] 
= N^{-1} \left[ \vAp(z), \vAm(z)  \right] 
=  \Lambda^2/(4 \pi^2 N)
\end{align}
we calculate
\begin{align}
\int d^3z \, \abs{\varphi}^2(z) \norm{\vAdm(z,t) \Psi_N}^2  
&= \int d^3z \abs{\varphi}^2(z) \SCP{\Psi_N}{\vAdp(z,t) \vAdm(z,t) \Psi_N}  
\nonumber \\
&= \int d^3z \, \abs{\varphi}^2(z)
\SCP{\Psi_N}{\left(\vAdm(z,t) \vAdp(z,t) + \Lambda^2/(4 \pi^2 N) \right)\Psi_N} 
\nonumber \\
&=  \int d^3z \, \abs{\varphi}^2(z) \norm{\vAdp(z,t) \Psi_N}^2 +  \Lambda^2/(4 \pi^2 N) 
\end{align}
and obtain
\begin{align}
\abs{\eqref{eq: beta-b-7}}   
\leq&  2 \beta^b + \norm{\kappa}_2^2  \Lambda^2/(\pi^2 N) + 8 \norm{\kappa}_2^2 \int d^3z \abs{\varphi}^2(z) \norm{\vAdp(z,t) \Psi_N}^2.
\end{align}
Then, we use \eqref{eq: Pauli relation A and E} and estimate
\begin{align}
\int d^3z \, \abs{\varphi}^2(z)  &\norm{\vAdp(z,t) \Psi_N}^2
=  \int d^3z \, \abs{\varphi}^2(z) \SCP{\Psi_N}{\vAdm(z,t) \vAdp(z,t) \Psi_N}  
\nonumber \\
&= \int d^3z \, \abs{\varphi}^2(z) \SCP{\Psi_N}{\int d^3y \, \eta(z-y) \vEdm(y,t) \int d^3l \, \eta(z-l) \vEdp(l,t) \Psi_N}  
\nonumber \\
&\leq \int d^3y \int d^3z \int d^3l \,
\abs{\varphi}^2(z) \abs{\eta(z-y)} \abs{\eta(z-l)}
\abs{\SCP{\vEdp(y,t) \Psi_N}{\vEdp(l,t) \Psi_N}} 
\nonumber \\
&\leq 1/2 \int d^3z \, \abs{\varphi}^2(z) \int d^3l \, \norm{\vEdp(l,t) \Psi_N}^2 \int d^3y \, \abs{\eta(z-y)}^2  
\nonumber \\
&+ 1/2 \int d^3z \, \abs{\varphi}^2(z) \int d^3y \, 
\norm{\vEdp(y,t) \Psi_N}^2 \int d^3l \, \abs{\eta(z-l)}^2 
\nonumber \\
&\leq \norm{\eta}_2^2 \int d^3y \, \norm{\vEdp(y,t) \Psi_N}^2 \leq \norm{\eta}_2^2 \beta^b
\leq C \Lambda \beta^b.
\end{align}
This yields
\begin{align}
\abs{\eqref{eq: beta-b-7}}   
\leq&  2 \beta^b + \norm{\kappa}_2^2 \Lambda^2/(\pi^2 N) + C \Lambda \norm{\kappa}_2^2  \beta^b
\leq C \Lambda^4 \left( \beta +   \Lambda/N \right).
\end{align}
The next terms of $d_t \beta(t)$ are bounded by
\begin{align}
\abs{\eqref{eq: beta-b-8}}
&\leq 4 \abs{\SCP{\kappa(y-x_1)q_1\Psi_N}{ i \nabla_1 p_1 \vEdp(y,t) \Psi_N}_{;y}}  
\nonumber \\
&\leq 2   \norm{\kappa(y-x_1) q_1 \Psi_N}_{;y}^2 
+ 2  \norm{i \nabla_1 p_1 \vEdp(y,t) \Psi_N}_{;y}^2 
\nonumber \\
&= 2 \SCP{q_1\Psi_N}{ \Big( \int d^3y \, \abs{\kappa(y-x_1)}^2 \Big) q_1\Psi_N} 
+2 \SCP{ \vEdp(y,t) \Psi_N}{p_1 \left(- \Delta_1\right) p_1  \vEdp(y,t) \Psi_N}_{;y} 
\nonumber \\
&= 2 \norm{\kappa}_2^2 \SCP{\Psi_N}{q_1 \Psi_N}
+ 2 \norm{\nabla \varphi}_2^2  \SCP{ \vEdp(y,t) \Psi_N}{ p_1  \vEdp(y,t) \Psi_N}_{;y} 
\nonumber \\
&\leq  2 \norm{\kappa}_2^2 \beta^a + 2 \norm{\nabla \varphi}_2^2 \SCP{\Psi_N}{ \vEdm(y,t) \vEdp(y,t) \Psi_N}_{;y}  
\leq  \Lambda^3  C(\norm{\varphi}_{H^2})  \beta
\end{align}
and
\begin{align}
\abs{\eqref{eq: beta-b-9}}
&\leq 4 \abs{\SCP{\kappa(y-z_1) N^{-1/2} \vA(x_1) q_1 \Psi_N}{p_1 \vEdp(y,t) \Psi_N}_{;y}} 
\nonumber \\
&\leq 2 \norm{\kappa(y-x_1)  N^{-1/2} \vA(x_1) q_1 \Psi_N}_{;y}^2  
+ 2   \norm{\vEdp(y,t) \Psi_N}_{;y}^2   
\nonumber \\
&= 2 \SCP{ N^{-1/2} \vA(x_1) q_1 \Psi_N}{ \Big( \int d^3y \, \abs{\kappa(y-x_1)}^2 \Big) N^{-1/2} \vA(x_1) q_1 \Psi_N} 
\nonumber \\
&+2  \SCP{ \Psi_N}{ \vEdm(y,t) \vEdp(y,t) \Psi_N}_{;y} 
\nonumber \\
&\leq 2 \beta^b + 2 \norm{\kappa}_2^2  \norm{  N^{-1/2} \vA(x_1) q_1 \Psi_N }^2
\nonumber \\
&\leq   \Lambda^4  C(\norm{\varphi}_{H^2},\norm{\varphi}_{\infty}, \mathcal{E}_M) \left( \beta + \Lambda/N  \right).
\end{align}
Here, we made use of Lemma~\ref{lemma: Pauli Feldoperator gegen alpha Abschaetzung}.
\begin{align}
\abs{\eqref{eq: beta-b-10}}
&\leq 4 \abs{\SCP{\Psi_N}{ \left( -i \nabla_1 - N^{-1/2} \vA(x_1) \right) \kappa(y-x_1) q_1 \vEdp(y,t) \Psi_N}_{;y}} 
\nonumber \\
&= 4 \abs{\SCP{N^{-1} \sum_{j=1}^N  q_j  \kappa(y - x_j) \left( -i \nabla_j -  N^{-1/2} \vA(x_j) \right) \Psi_N}{ \vEdp(y,t) \Psi_N}_{;y}} 
\nonumber \\
&\leq 2 \Big| \Big| N^{-1} \sum_{j=1}^N q_j  \kappa(y-x_j) \left( -i \nabla_j -  N^{-1/2} \vA(x_j) \right) \Big| \Big|_{;y}^2  
+ 2  \norm{ \vEdp(y,t)  \Psi_N }_{;y}^2.
\end{align}
According to Lemma~\ref{lemma: Pauli field operator times p or q against functional} and Lemma~\ref{lemma: crucial bound} this is bounded by
\begin{align}
\abs{\eqref{eq: beta-b-10}}
\leq& \Lambda^3 C(\norm{\varphi}_{H^2}, \norm{\varphi}_{\infty}, \mathcal{E}_M) \left(  \beta + \Lambda/N  \right).
\end{align}
}

\subsection{Bound on $d_t \beta$:}
The Maxwell-Schr\"odinger equations are a conserved system and its energy does not change during the time evolution
\begin{align}
\mathcal{E}_M[\varphi_t, u_t] = \mathcal{E}_M[\varphi_0, u_0].
\end{align}
Moreover, $\beta^c$ is a constant of motion because the self-adjointness of the Pauli-Fierz Hamiltonian gives rise to a strongly continuous unitary group $\{e^{-i t H_N} \}_{t \in \mathbb{R}}$ such that $\Psi_{N,t} = e^{-i t H_N} \Psi_{N,0}$ and 
\begin{align}
\label{eq: Pauli time dependence of beta-c}
\beta^c( \Psi_{N,t},\varphi_t,  u_t ) &= \norm{\left( N^{-1} H_N - \mathcal{E}_M[\varphi_t, u_t] \right) \Psi_{N,t}}^2  
\nonumber\\
&= \norm{\left(  N^{-1} H_N  - \mathcal{E}_M[\varphi_0, u_0] \right) e^{-i t H_N} \Psi_{N,0}}^2  
\nonumber \\
&=  \norm{e^{-i t H_N} \left( N^{-1} H_N - \mathcal{E}_M[\varphi_0, u_0] \right) \Psi_{N,0}}^2  
= \beta^c (\Psi_{N,0},\varphi_0, u_0 ).
\end{align}
This allows us to bound the time derivative of $\beta(t)$ by
\begin{lemma}
\label{lemma: Pauli dt beta}
Let $v$ satisfy (A1), $(\varphi_t,\boldsymbol{A}(t), \boldsymbol{E}(t)) \in \mathcal{G}$ and  $\Psi_{N,t}$ be the unique solution of \eqref{eq: Pauli Schroedinger equation microscopic} with initial data $\Psi_{N,0} \in \left(L^2_s \left( \mathbb{R}^{3N} \right) \otimes \mathcal{F}_p \right) \cap \mathcal{D}(H_N)$. 
Then, there exists a monotone increasing function $C(t)$ of $\norm{\vAc}_{\infty}$, $\mathcal{E}_M[\varphi_t, u_t]$, $\norm{\varphi_t}_{H^2(\mathbb{R}^3)}$ and $\norm{\varphi_t}_{L^{\infty}(\mathbb{R}^3)}$  such that
\begin{align}
\abs{d_t \beta (\Psi_{N,t},\varphi_t, u_t)} &\leq  \Lambda^4  C(t)  \left( \beta(\Psi_{N,t},\varphi_t, u_t) + \Lambda/N \right) , 
\nonumber \\
\beta(\Psi_{N,t},\varphi_t, u_t) &\leq  e^{\Lambda^4 \int_0^t ds C(s)} \left( \beta(\Psi_{N,0},\varphi_0, u_0) + \Lambda/N  \right)
\end{align}
holds for any $t \geq 0$.
\end{lemma}

\begin{proof}
The first inequality is a direct consequence of  Lemma~\ref{lemma: dt beta-a}, Lemma~\ref{lemma: dt beta-b} and \eqref{eq: Pauli time dependence of beta-c}.
Then, we apply Gronwall's inequality and obtain
\begin{align}
\beta(\Psi_{N,t},\varphi_t, u_t) &\leq e^{\Lambda^4 \int_0^t ds C(s)} \left( \beta(\Psi_{N,0},\varphi_0, u_0) + \Lambda/N \right).
\end{align}
\end{proof}

\section{Initial conditions}
\label{section: initial conditions}

In this section, we show that $\beta(\Psi_{N,0}, \varphi_0, u_0)$ is small for the initial states of Theorem~\ref{theorem: Pauli main theorem}.
\begin{lemma}
\label{lemma: Pauli smallness of initial beta}
Let $\Psi_{N,0} \in \mathcal{D}(H_N) \cap \left( L_s^2(\mathbb{R}^{3N}) \otimes \mathcal{F}_p  \right)$, $\varphi_0 \in H^2(\mathbb{R}^3)$ with $\norm{\varphi_0} = 1$ and $\alpha_0 \in \mathfrak{h}$ such that $A(x,0) \in H^2(\mathbb{R}^3, \mathbb{C}^3)$ and  $E(x,0) \in H^1(\mathbb{R}^3, \mathbb{C}^3)$. Then
\begin{align}
\label{eq: Pauli initial beta-a}
\beta^a(\Psi_{N,0}, \varphi_0) &\leq \text{Tr}_{L^2(\mathbb{R}^3)} \abs{\gamma_{N,0}^{(1,0)} - \ket{\varphi_0} \bra{\varphi_0} } = a_N ,
 \\
 \label{eq: Pauli initial beta-b}
\beta^b(\Psi_{N,0}, u_0) &= N^{-1} \scp{W^{-1}(\sqrt{N} \alpha_0)\Psi_{N,0}}{H_f W^{-1}(\sqrt{N} \alpha_0)\Psi_{N,0}} = b_N ,
\\
\label{eq: Pauli initial beta-c}
\beta^c(\Psi_{N,0},\varphi_0, u_0) &= c_N.
\end{align}
In particular for $\Psi_{N,0} = \varphi_0^{\otimes N} \otimes W(\sqrt{N} \alpha_0) \Omega$ (and $\Lambda/N \leq  C$)  we have
\begin{align}
\beta (\Psi_{N0},\varphi_0, u_0) \leq C \Lambda^4 N^{-1}   .
\end{align}
\end{lemma}
\noindent
Before we prove Lemma~\ref{lemma: Pauli smallness of initial beta} we recall some  well known properties of Weyl operators \eqref{eq: Pauli Weyl operator}.

\begin{lemma}
\label{lemma: Pauli Weyl operator properties}
Let
$f, g \in \mathfrak{h} = L^2(\mathbb{R}^3) \otimes \mathbb{C}^2$.
\begin{itemize}
\item[(i)] $W(f)$ is a unitary operator and 
\begin{align}
W^*(f) = W^{-1}(f) = W(-f).
\end{align}
\item[(ii)] We have
\begin{align}
W^*(f)a(k,\lambda) W(f) &= a(k,\lambda) + f(k,\lambda),
\nonumber \\
W^*(f)a^*(k,\lambda) W(f) &=  a^*(k,\lambda) + f^*(k,\lambda).
\end{align}
\item[(iii)] From (ii) we see that coherent states are eigenvectors of annihilation operators
\begin{align}
a(k,\lambda) W(f) \Omega = f(k,\lambda) W(f) \Omega .
\end{align}
\end{itemize}
\end{lemma} 
\noindent
Subsequently, we compute the expectation values of the vector potential, the field energy and higher moments. 
\begin{lemma}
\label{lemma: Pauli mean values for Weyl states}
Let  $\alpha_0 \in \mathfrak{h} $ such that $u_0 \in \mathfrak{h}$ and $\vEpc(x,t)$ be defined by  \eqref{eq: Pauli positive frequency part of the classical electric field}. Moreover, let
\begin{align}
\vgp_{il}(x) &\coloneqq \int d^3k \; \abs{\tilde{\kappa}(k)}^2 \abs{k}^{-1}
e^{ikx} \left( \delta_{il} -   k_i k_l \abs{k}^{-2} \right)
\quad \text{with}  \quad
||\vgp_{il}||_2^2 \leq   \Lambda.
\end{align}
Then
\begin{align}
\scp{W(\sqrt{N} \alpha_0) \Omega}{ N^{-1/2} \vA(x) W(\sqrt{N} \alpha_0) \Omega}_{\mathcal{F}_p} 
&= \vAc(x,0) ,    
\nonumber \\
\scp{W(\sqrt{N} \alpha_0) \Omega}{N^{-1} \vA^2(x) W(\sqrt{N} \alpha_0) \Omega}_{\mathcal{F}_p}
&= \vAc^2(x,0) + \Lambda^2/(4 \pi^2 N), 
\nonumber \\
\scp{W(\sqrt{N} \alpha_0) \Omega}{N^{-1} \vA^i(x) \vA^j(y) W(\sqrt{N} \alpha_0) \Omega}_{\mathcal{F}_p}  
&= \vAc^i(x,0) \vAc^j(y,0) + (2N)^{-1} \vgp_{ij}(x-y)    ,  
\nonumber \\
\scp{W(\sqrt{N} \alpha_0) \Omega}{ N^{-1} H_f W(\sqrt{N} \alpha_0) \Omega}_{\mathcal{F}_p}
&= \norm{u_0}_{\mathfrak{h}}^2,   
\nonumber \\
\scp{W(\sqrt{N} \alpha_0) \Omega}{  N^{-2} H_f^2 W(\sqrt{N} \alpha_0) \Omega}_{\mathcal{F}_p}
&= \norm{u_0}_{\mathfrak{h}}^4 + N^{-1}  ||\abs{\cdot}^{1/2} u_0 ||_{\mathfrak{h}}^2,    
\nonumber \\
\scp{W(\sqrt{N} \alpha_0) \Omega}{  N^{-3/2} \vA(x) H_f  W(\sqrt{N} \alpha_0) \Omega}_{\mathcal{F}_p}
&= \vAc(x,0)  \norm{u_0}_{\mathfrak{h}}^2 -  i N^{-1} \vEpc(x,0)  ,
\nonumber \\
\scp{W(\sqrt{N} \alpha_0) \Omega}{ N^{-2} \vA^2(x) H_f  W(\sqrt{N} \alpha_0) \Omega}_{\mathcal{F}_p}
&=  \left( \vAc^2(x,0)  +   \Lambda^2/(4 \pi^2 N) \right) \norm{u_0}_{\mathfrak{h}}^2
\nonumber \\
& - 2 i N^{-1}  \vAc(x,0) \vEpc(x,0) .
\end{align}
\end{lemma}

\begin{proof}
The proof is a simple application of the canonical commutation relations~\eqref{eq: canonical commutation relation} and  part (ii) from Lemma~\ref{lemma: Pauli Weyl operator properties}.
\end{proof}

\begin{proof}[Proof of Lemma~\ref{lemma: Pauli smallness of initial beta}]
Relation~\eqref{eq: Pauli initial beta-a} directly follows from Lemma~\ref{lemma: Pauli relation between beta and reduced density matrices}.
In view of Lemma~\ref{lemma: Pauli Weyl operator properties} we calculate
\begin{align}
\label{eq: Pauli relation betweeen beta-b and b-N}
\beta^b&(\Psi_{N,0}, u_0) 
= \sum_{\lambda=1,2} \int d^3k \, \abs{k} \norm{\left( N^{-1/2} a(k,\lambda) - \alpha_0(k,\lambda) \right) \Psi_{N,0}}^2
\nonumber \\
&= \sum_{\lambda=1,2} \int d^3k \, \abs{k} \norm{W^{-1}(\sqrt{N} \alpha_0)  \left( N^{-1/2} a(k,\lambda) - \alpha_0(k,\lambda) \right) W(\sqrt{N} \alpha_0) W^{-1}(\sqrt{N} \alpha_0)  \Psi_{N,0}}^2
\nonumber \\
&= \sum_{\lambda=1,2} \int d^3k \, \norm{N^{-1/2} a(k,\lambda) W^{-1}(\sqrt{N} \alpha_0) \Psi_{N,0}}^2
\nonumber \\
&= N^{-1} \SCP{W^{-1}(\sqrt{N} \alpha_0) \Psi_{N,0}}{H_f W^{-1}(\sqrt{N} \alpha_0) \Psi_{N,0}} = b_N.
\end{align}
Equation~\eqref{eq: Pauli initial beta-c} is solely the definition of $\beta^c$.
In the following, we are interested in initial data $\Psi_{N,0} = \varphi_0^{\otimes N} \otimes W(\sqrt{N} \alpha_0) \Omega$ of  product type. First, we notice that
\begin{align}
\beta^a(\Psi_{N,0},\varphi_0)
=& \SCP{\Psi_{N,0}}{q_1^{\varphi_0} \Psi_{N,0}}   
= \scp{\varphi_0}{\varphi_0}_{L^2(\mathbb{R}^3)} - \scp{\varphi_0}{\varphi_0}_{L^2(\mathbb{R}^3)}^2 = 0
\end{align}
because the scalar product factorizes for product states and $q_1$ only acts on the Hilbert space of the first charged particle. Then, we follow
\begin{align}
\beta^b[\Psi_{N,0}, u_0] = 0
\end{align}
for $\Psi_{N,0} = \varphi_0^{\otimes N} \otimes W(\sqrt{N} \alpha_0) \Omega$ from \eqref{eq: Pauli relation betweeen beta-b and b-N}.
To show that the product structure suppresses the fluctuations of the energy per particle around its mean value is more elaborate. Nevertheless, the idea of the proof simple and in the spirit of the law of large numbers from probability theory.
We bound $\beta^c$ by
\begin{align}
\beta^c(0) 
&= \SCP{\left( N^{-1} H_N - \mathcal{E}_M \right) \Psi_{N,0}}{\left( N^{-1} H_N - \mathcal{E}_M \right) \Psi_{N,0}}  
\nonumber \\
&\leq \abs{ \SCP{ N^{-1} H_N \Psi_{N,0}}{ N^{-1} H_N \Psi_{N,0}}
- \mathcal{E}_M^2 }
+ 2 \mathcal{E}_M \abs{ \mathcal{E}_M 
- \SCP{\Psi_{N,0}}{ N^{-1} H_N \Psi_{N,0}}  }
\end{align}
and show that
\begin{itemize}
\item[(i)] $\abs{\SCP{\Psi_{N,0}}{ N^{-1} H_N \Psi_{N,0}} - \mathcal{E}_M[\varphi_0, u_0] } \leq  C \Lambda^2  N^{-1} $
\item[(ii)] $\abs{\SCP{ N^{-1} H_N \Psi_{N,0}}{ N^{-1} H_N \Psi_{N,0}} -\mathcal{E}^2_M[\varphi_0, u_0]} \leq C \Lambda^4 N^{-1} $
\end{itemize}
holds for states of product type.
 
\subsubsection*{(i) The mean value of the energy per particle}
For ease of notation we denote $\vAc(\cdot,0)$, $\vEpc(\cdot,0)$, $\norm{u_0}_{\mathfrak{h}}^2$, $||\abs{\cdot}^{1/2} u_0 ||_{\mathfrak{h}}^2$ by $\vAc(\cdot)$, $\vEc(\cdot)$, $\norm{u_0}^2$, $||\abs{\cdot}^{1/2} u_0 ||^2$ in the following. 
The mean value of the energy per particle is given by
\begin{align}
\SCP{\Psi_{N,0}}{ N^{-1} H_N \Psi_{N,0}}
&= \SCP{\Psi_{N,0}}{ N^{-1} \sum_{j=1}^N  \left( - i  \nabla_j -  N^{-1/2} \vA(x_j) \right)^2  \Psi_{N,0}} 
\nonumber \\
&+  \SCP{\Psi_{N,0}}{1/(2 N^2) \sum_{j \neq k} v(x_j -x_k) \Psi_{N,0}} 
\nonumber  \\
&+  \SCP{\Psi_{N,0}}{ N^{-1} H_f \Psi_{N,0}}.
\end{align}
Due to symmetry and the product structure of $\Psi_{N,0}$ this becomes
\begin{align}
\SCP{\Psi_{N,0}}{ N^{-1} H_N \Psi_{N,0}}
&= \scp{\varphi_0}{\left( - \Delta \right) \varphi_0} 
\nonumber \\
&+ 2 i \scp{\varphi_0}{\scp{W(\sqrt{N} \alpha_0) \Omega}{ N^{-1/2} \vA W(\sqrt{N} \alpha_0) \Omega}_{\mathcal{F}_p} \cdot \nabla \varphi_0} 
\nonumber \\
&+ \scp{\varphi_0}{\scp{W(\sqrt{N} \alpha_0) \Omega}{ N^{-1} \vA^2 W(\sqrt{N} \alpha_0) \Omega}_{\mathcal{F}_p}  \varphi_0} 
\nonumber \\
&+ (N-1)/(2 N) \scp{\varphi_0}{\left( v * \abs{\varphi_0}^2 \right) \varphi_0} 
\nonumber  \\
&+  \scp{W(\sqrt{N} \alpha_0) \Omega}{ N^{-1} H_f W(\sqrt{N} \alpha_0) \Omega}_{\mathcal{F}_p}  .
\end{align}
Lemma~\ref{lemma: Pauli mean values for Weyl states}  gives
\begin{align}
\SCP{\Psi_{N,0}}{ N^{-1} H_N \Psi_{N,0}}
&= \norm{\left(- i \nabla - \vAc \right) \varphi_0}^2  
+ 1/2 \scp{\varphi_0}{\left( v * \abs{\varphi_0}^2 \right) \varphi_0} 
+  \norm{u_0}^2
\nonumber  \\
&+ \Lambda^2/(4 \pi^2 N) 
- 1/(2N) \scp{\varphi_0}{\left(v * \abs{\varphi_0}^2 \right) \varphi_0}. 
\end{align}
Now we can pull the sums together to get
\begin{align}
\SCP{\Psi_{N,0}}{ N^{-1} H_N \Psi_{N,0}}
&= \mathcal{E}_M[\varphi_0, u_0]
+ \Lambda^2/(4 \pi^2 N) 
- 1/(2N) \scp{\varphi_0}{\left(v * \abs{\varphi_0}^2 \right) \varphi_0},
\end{align}
where  $\scp{\varphi_0}{\left(v * \abs{\varphi_0}^2 \right) \varphi_0} \leq C(\norm{\varphi_0}_{\infty})$ holds according to \eqref{eq: Pauli direct interaction potential Young 2}.

\subsubsection*{(ii) The second moment of the energy per particle}
Subsequently, we show that the second moment of the energy per particle approximately equals the energy of the effective system squared.
We split the double sum, arising from the second moment of the many-body Hamiltonian into its diagonal and off-diagonal part. 
The diagonal only consists of $N$ constituents and has a subleading contribution for large $N$. On the contrary, there are $N^2$ elements from the off-diagonal which give rise to $\mathcal{E}_M^2$. In order to organize the estimate, we first decompose the second moment of the energy per particle as well as the effective energy squared into pieces:
\begin{align}
\SCP{ N^{-1} & H_N\Psi_{N,0}}{ N^{-1} H_N \Psi_{N,0}} =   
\nonumber \\
\label{eq: microscopic energy squared 1}
=& N^{-2} \sum_{j,k}
\SCP{ \left(- i \nabla_j - N^{-1/2} \vA(x_j)   \right)^2   \Psi_{N,0}}{ \left(- i \nabla_k - N^{-1/2} \vA(x_k)   \right)^2  \Psi_{N,0}}  \\
\label{eq: microscopic energy squared 2}
+&  ( 4 N^{4})^{-1} \sum_{i \neq j, k \neq l}
\SCP{v(x_i -x_j) \Psi_{N,0}}{v(x_k -x_l) \Psi_{N,0}}  \\
\label{eq: microscopic energy squared 3}
+& N^{-2} \SCP{\Psi_{N,0}}{H_f^2 \Psi_{N,0}}  \\
\label{eq: microscopic energy squared 4}
+& N^{-3} \sum_{j,k \neq l} \Re
\SCP{ \left(- i \nabla_j -  N^{-1/2} \vA(x_j)   \right)^2   \Psi_{N,0}}{v(x_k - x_l) \Psi_{N,0}}   \\
\label{eq: microscopic energy squared 5}
+& 2N^{-2} \sum_{j}
\Re \SCP{\left(- i \nabla_j -  N^{-1/2} \vA(x_j)   \right)^2  \Psi_{N,0}}{ H_f \Psi_{N,0}}   \\
\label{eq: microscopic energy squared 6}
+& N^{-3} \sum_{ j \neq k} \Re 
\SCP{v(x_j - x_k) \Psi_{N,0}}{ H_f \Psi_{N,0}}
\end{align}
and
\begin{align}
\label{eq: macroscopic energy squared 1}
\mathcal{E}_M^2[\varphi_0,u_0] 
=& \scp{\varphi_0}{ \left( - i \nabla - \vAc \right)^2    \varphi_0}^2   \\
\label{eq: macroscopic energy squared 2}
+& 1/4 \scp{\varphi_0}{(v * \abs{\varphi_0}^2) \varphi_0}^2    \\
\label{eq: macroscopic energy squared 3}
+& \norm{u_0}^4  \\
\label{eq: macroscopic energy squared 4}
+& \scp{\varphi_0}{\left( \left( - i \nabla - \vAc \right)^2   \right) \varphi_0} \scp{\varphi_0}{(v * \abs{\varphi_0}^2) \varphi_0}   \\
\label{eq: macroscopic energy squared 5}
+& 2 \scp{\varphi_0}{ \left( - i \nabla - \vAc \right)^2   \varphi_0} \norm{u_0}^2  \\
\label{eq: macroscopic energy squared 6}
+& \scp{\varphi_0}{(v * \abs{\varphi_0}^2) \varphi_0} \norm{u_0}^2.
\end{align}
Then, we estimate the difference between the corresponding expressions and obtain 
\begin{align}
\abs{\SCP{ N^{-1} H_N \Psi_{N,0}}{ N^{-1} H_N \Psi_{N,0}} -\mathcal{E}^2_M[\varphi_0,u_0]} \leq  C \Lambda^4 N^{-1} .
\end{align}

\subsubsection*{ $\abs{ \eqref{eq: microscopic energy squared 1} - \eqref{eq: macroscopic energy squared 1}} \leq C \Lambda^4/N $: }

The off-diagonal part of \eqref{eq: microscopic energy squared 1} is given by
\begin{align}
\label{eq: microscopic energy squared 1 c off diagonal}
\SCP{ \left(- i \nabla_1 -  N^{-1/2} \vA(x_1)   \right)^2 &\Psi_{N,0}}{\left(- i \nabla_2 -  N^{-1/2} \vA(x_2)   \right)^2 \Psi_{N,0}}  \\
\label{eq: microscopic energy squared 1 c 1}
&= \SCP{\left( - \Delta_1 \right) \Psi_{N,0}}{\left( - \Delta_2 \right) \Psi_{N,0}}   \\
\label{eq: microscopic energy squared 1 c 2a}
&+ 2i  \SCP{\left( - \Delta_1 \right) \Psi_{N,0}}{ N^{-1/2} \vA(x_2) \nabla_2 \Psi_{N,0}}   
\\
\label{eq: microscopic energy squared 1 c 2b}
&- 2i  \SCP{  N^{-1/2} \vA(x_1) \nabla_1 \Psi_{N,0}}{\left( - \Delta_2 \right) \Psi_{N,0}}    \\
\label{eq: microscopic energy squared 1 c 3a}
&+  \SCP{\left( - \Delta_1 \right) \Psi_{N,0}}{ N^{-1} \vA^2(x_2) \Psi_{N,0}} \\
\label{eq: microscopic energy squared 1 c 3b}
&+   \SCP{  N^{-1} \vA^2(x_1) \Psi_{N,0}}{\left( - \Delta_2 \right) \Psi_{N,0}}  \\
\label{eq: microscopic energy squared 1 c 4}
&+ 4 \SCP{ N^{-1/2} \vA(x_1) \nabla_1 \Psi_{N,0}}{ N^{-1/2} \vA(x_2) \nabla_2 \Psi_{N,0}}   \\
\label{eq: microscopic energy squared 1 c 5a}
&- 2i  \SCP{ N^{-1/2} \vA(x_1) \nabla_1 \Psi_{N,0}}{ N^{-1} \vA^2(x_2)  \Psi_{N,0}}  \\
 \label{eq: microscopic energy squared 1 c 5b} 
&+ 2i \SCP{ N^{-1} \vA^2(x_1) \Psi_{N,0}}{ N^{-1/2} \vA(x_2) \nabla_2 \Psi_{N,0}}     \\
\label{eq: microscopic energy squared 1 c 6}
&+ \SCP{ N^{-1} \vA^2(x_1) \Psi_{N,0}}{ N^{-1} \vA^2(x_2) \Psi_{N,0}}.
\end{align}

\noindent
By means of Lemma~\ref{lemma: Pauli mean values for Weyl states}  we have
\begin{align}
\eqref{eq: microscopic energy squared 1 c 1}
+  \eqref{eq: microscopic energy squared 1 c 2a} + \eqref{eq: microscopic energy squared 1 c 2b}
=& \scp{\varphi_0}{\left( - \Delta \right) \varphi_0}^2
+ 4 i \scp{\vAc \varphi_0}{ \nabla \varphi_0} 
\scp{\varphi_0}{\left( - \Delta \right) \varphi_0}, 
\nonumber \\
\eqref{eq: microscopic energy squared 1 c 3a} + \eqref{eq: microscopic energy squared 1 c 3b}
=& 2 \scp{ \varphi_0}{\vAc^2 \varphi_0} 
\scp{\varphi_0}{\left( - \Delta \right) \varphi_0}
+ \Lambda^2/(2 \pi^2 N) \norm{\nabla \varphi_0}^2,
\nonumber  \\
\eqref{eq: microscopic energy squared 1 c 4}
=& - 4 \scp{\varphi_0}{\vAc \nabla \varphi_0}^2  - 
\nonumber \\
-& 2/N
 \int d^3x \int d^3y \varphi_0^*(x) \varphi_0^*(y) \vgp_{kl}(x-y)
 \nabla^k \varphi_0(x) \nabla^l \varphi_0(y).  
 \end{align}
In order to evaluate the last three lines, we use that
\begin{align}
 N^{-3/2} \scp{W(\sqrt{N} \alpha_0) \Omega}{\vA^2(x) \vA^i(y) W(\sqrt{N} \alpha_0) \Omega}_{\mathcal{F}_p}  
=& \vAc^2(x) \vAc^i(y)   
+   \Lambda^2/(4 \pi^2 N) \vAc^i(y)  
\nonumber \\
+& N^{-1} \sum_{j=1}^3  \vgp_{ij}(x-y) \vA^j(x)    ,   
\end{align}
and 
\begin{align}
N^{-2}  \scp{W(\sqrt{N} \alpha_0) \Omega}{&\vA^2(x) \vA^2(y) W(\sqrt{N} \alpha_0) \Omega}_{\mathcal{F}_p} 
= \vAc^2(x) \vAc^2(y)  +   
\nonumber \\
&+  \Lambda^2/(4 \pi^2 N) \left( \vAc^2(x) + \vAc^2(y) \right)
+ 2/N \sum_{k,l=1}^3 \vgp_{kl}(x-y) \vAc^k(x) \vAc^l(y)  
\nonumber \\
&+ N^{-2} \Big( \sum_{k,l=1}^3 \abs{\vgp_{kl}(x-y)}^2 + \Lambda^4/(2 \pi)^4 \Big),
\end{align}
can also be obtained by the canonical commutation relations~\eqref{eq: canonical commutation relation} and  Lemma~\ref{lemma: Pauli Weyl operator properties}. Consequently, we have
\begin{align}
\eqref{eq: microscopic energy squared 1 c 5a} + \eqref{eq: microscopic energy squared 1 c 5b}
&= 4 i \scp{\varphi_0}{\vAc^2 \varphi_0}  \scp{\varphi_0}{\vAc \nabla \varphi_0}  
+ i \Lambda^2/(\pi^2 N) \scp{\varphi_0}{\vAc \nabla \varphi_0} 
\nonumber \\
&+ 4i/N \int d^3x \int d^3y  \, \varphi^*_0(x) \varphi^*_0(y)
\vgp_{kl}(x-y) \vAc^k(x) \varphi(x) \big( \nabla^l \varphi \big)(y) ,     
 \nonumber \\
\eqref{eq: microscopic energy squared 1 c 6}
&= \scp{\varphi_0}{\vAc^2 \varphi_0}^2
+ \Lambda^4/(2 \pi^2 N)  \scp{\varphi_0}{\vAc^2 \varphi_0} + \Lambda^4/(16 \pi^4 N^2)
\nonumber \\
&+   N^{-2}  \int d^3x   \int d^3y \, \varphi_0^*(x) \varphi_0^*(y)
\sum_{k,l} \abs{\vgp_{k,l}(x-y)}^2 \varphi(x) \varphi(y)
\nonumber  \\
&+ 2/N \int d^3x \int d^3y  \, \varphi_0^*(x) \varphi_0^*(y) \vgp_{kl}(x-y)
\vAc^k(x) \vAc^l(y) \varphi(x) \varphi(y).
\end{align}
Pulling all the pieces together and using that all error terms are bounded by $C \Lambda^4/N$ under the assumptions of Lemma~\ref{lemma: Pauli smallness of initial beta}, gives
\begin{align}
\abs{  \eqref{eq: microscopic energy squared 1 c off diagonal}
- \scp{\varphi_0}{\left( - i \nabla - \vAc \right)^2 \varphi_0}^2}
\leq& C \Lambda^4/N.
\end{align}
Since the diagonal part of \eqref{eq: microscopic energy squared 1} is of order $N^{-1}$, this implies
\begin{align}
\abs{ \eqref{eq: microscopic energy squared 1} - \eqref{eq: macroscopic energy squared 1}} \leq C \Lambda^4/N .
\end{align}

\subsubsection*{$\abs{ \eqref{eq: microscopic energy squared 2} - \eqref{eq: macroscopic energy squared 2}} \leq C/N $: }
By virtue of the symmetry of the wave function and $v(-x) = v(x)$ we can write line~\eqref{eq: microscopic energy squared 2} as
\begin{align}
(4 N^4)^{-1} \sum_{i \neq j, k \neq l}
&\SCP{v(x_i -x_j) \Psi_{N,0}}{v(x_k -x_l) \Psi_{N,0}}  =
\nonumber \\
&= 1/4 \SCP{v(x_1 - x_2) \Psi_{N,0}}{v(x_3 - x_4) \Psi_{N,0}}     
\nonumber \\
&- (6N^2 - 11 N + 6)N^{-3} \SCP{v(x_1 - x_2) \Psi_{N,0}}{v(x_3 - x_4) \Psi_{N,0}} 
\nonumber \\
&+ (N-1) N^{-3}/2 \SCP{v(x_1 - x_2) \Psi_{N,0}}{v(x_1 - x_2) \Psi_{N,0}} 
\nonumber \\
&+ (N-1)(N-2)N^{-3} \SCP{v(x_1 - x_2) \Psi_{N,0}}{v(x_1 - x_3) \Psi_{N,0}}.
\end{align}
The product structure of the initial state, \eqref{eq: Pauli direct interaction potential Young 1} and \eqref{eq: Pauli direct interaction potential Young 2} give
\begin{align}
\SCP{v(x_1 - x_2) \Psi_{N,0}}{v(x_3 - x_4) \Psi_{N,0}} &= \scp{\varphi_0}{\left( v * \abs{\varphi_0}^2 \right) \varphi_0}^2  \leq C( \norm{\varphi_0}_{\infty}) ,  
\nonumber \\
\norm{v(x_1-x_2) \Psi_{N,0}}^2 &= \scp{\varphi_0}{\left( v^2 * \abs{\varphi_0}^2 \right) \varphi_0}
\leq C( \norm{\varphi_0}_{\infty})
\end{align}
and we conclude 
\begin{align}
\abs{ (4 N^4)^{-1} \sum_{i \neq j, k \neq l}
&\SCP{v(x_i -x_j) \Psi_{N,0}}{v(x_k -x_l) \Psi_{N,0}}  -  1/4 \scp{\varphi_0}{\left( v * \abs{\varphi_0}^2 \right) \varphi_0}^2  } \leq
\nonumber \\
&\leq 6/N \abs{\SCP{v(x_1 - x_2) \Psi_{N,0}}{v(x_3 - x_4) \Psi_{N,0}}} 
+  N^{-1}  \norm{v(x_1-x_2) \Psi_{N0}}^2   
\nonumber \\
&+ N^{-1} \abs{\SCP{v(x_1 - x_2) \Psi_{N,0}}{v(x_1 - x_3) \Psi_{N,0}}}  
\nonumber \\
&\leq 8/N  \norm{v(x_1-x_2) \Psi_{N,0}}^2
= 8/N \scp{\varphi_0}{\left( v^2 *\abs{\varphi_0}^2 \right) \varphi_0} .
\end{align}

\subsubsection*{$\abs{ \eqref{eq: microscopic energy squared 3} - \eqref{eq: macroscopic energy squared 3}} \leq C/N $:}
This bound results from Lemma~\ref{lemma: Pauli mean values for Weyl states} because
\begin{align}
N^{-2} \SCP{\Psi_{N0}}{H_f^2 \Psi_{N0}}
&= N^{-2} \scp{W(\sqrt{N} \alpha_0) \Omega}{H_f^2 W(\sqrt{N} \alpha_0) \Omega}_{\mathcal{F}_p}  
= \norm{u_0}^4 + N^{-1} ||\abs{\cdot}^{1/2} u_0 ||^2 .
\end{align}

\subsubsection*{$\abs{ \eqref{eq: microscopic energy squared 4} - \eqref{eq: macroscopic energy squared 4}} \leq C \Lambda^2/N$: }
Line \eqref{eq: microscopic energy squared 4} simplifies to
\begin{align}
 N^{-3} \sum_{j,k \neq l} \Re
&\SCP{ \left(- i \nabla_j -  N^{-1/2} \vA(x_j)   \right)^2   \Psi_{N,0}}{v(x_k - x_l) \Psi_{N,0}}  
\nonumber \\
&= (N-1)(N-2) N^{-2} \Re 
\SCP{ \left(- i \nabla_1 -  N^{-1/2} \vA(x_1)   \right)^2   \Psi_{N,0}}{v(x_2 - x_3) \Psi_{N,0}}    
\nonumber \\
&+ 2(N-1) N^{-2} \Re
\SCP{ \left(- i \nabla_1 - N^{-1/2} \vA(x_1)   \right)^2   \Psi_{N,0}}{v(x_2 - x_1) \Psi_{N,0}}   
\nonumber  \\
&= \left( 1 - 3(N-2) N^{-2} \right) \scp{\varphi_0}{ \left( -i \nabla - \vAc \right)^2   \varphi_0}
\scp{\varphi_0}{\left( v * \abs{\varphi_0}^2 \right) \varphi_0}  
\nonumber \\
&+ (N-1)(N-2) N^{-3} \Lambda^2/(4 \pi^2) \scp{\varphi_0}{\left( v * \abs{\varphi_0}^2 \right) \varphi_0}  
\nonumber  \\
&+ 2(N-1) N^{-2} \Re
\SCP{ \left(- i \nabla_1 - N^{-1/2} \vA(x_1)   \right)^2   \Psi_{N,0}}{v(x_2 - x_1) \Psi_{N,0}}.
\end{align}
Consequently the estimate follows  because  $\Big| \Big| \left(- i \nabla_1 -  N^{-1/2} \vA(x_1)   \right)^2 \Psi_{N,0} \Big| \Big|$ is finite under the assumptions of Lemma~\ref{lemma: Pauli smallness of initial beta}.

\subsubsection*{$\abs{ \eqref{eq: microscopic energy squared 5} - \eqref{eq: macroscopic energy squared 5}} \leq C \Lambda^2/N $: }
Similar to the previous calculations we obtain
\begin{align}
2N^{-2} \sum_{j=1}^N \Re
&\SCP{ \left( - i \nabla_j - N^{-1/2} \vA(x_j) \right)^2   \Psi_{N,0}}{H_f \Psi_{N,0}} 
\nonumber \\
=& 2 \Re \SCP{\left( - \Delta_1 + 2i N^{-1/2}  \vA(x_1) + N^{-1} \vA^2(x_1)  \right) \Psi_{N,0}}{ N^{-1} H_f \Psi_{N,0}} 
\nonumber \\
=& 2 \scp{\varphi_0}{ \left( -i \nabla - \vAc \right)^2  \varphi_0} \norm{u_0}^2
\nonumber \\
+& 2 N^{-1} \Re \left(  \Lambda^2/(4 \pi^2) \norm{u_0}^2
- 4 \scp{\nabla \varphi_0}{\vEc^+(x) \varphi_0} - 2i \scp{\vAc \varphi_0}{\vEc^+ \varphi_0}  \right).
\end{align}
By means of
\begin{align}
\abs{\scp{\vAc \varphi_0}{\vEc^+ \varphi_0}} \leq& \norm{\varphi_0}_{\infty} \norm{\vEc^+}  \norm{\vAc}_{\infty},
\nonumber \\
\abs{\scp{\nabla \varphi_0}{\vEc^+ \varphi_0}} \leq& \norm{\varphi_0}_{\infty} \norm{\vEc^+} \norm{\nabla \varphi_0} ,
\end{align}
and 
\begin{align}
\norm{\vEc^+}_2^2 
=& \frac{1}{2} \sum_{\lambda=1,2} \int_{\abs{k} \leq \Lambda} d^3k \, \abs{k} \abs{\alpha_0(k,\lambda)}^2 \leq \norm{u_0}^2
\end{align}
the inequality follows.

\subsubsection*{$\abs{ \eqref{eq: microscopic energy squared 6} - \eqref{eq: macroscopic energy squared 6}} \leq C/N $:}
Making use of symmetry and Lemma~\ref{lemma: Pauli mean values for Weyl states} one has
\begin{align}
N^{-3} \sum_{j \neq k} \Re \SCP{v(x_k - x_k) \Psi_{N,0}}{H_f \Psi_{N,0}}  
&= (N-1) N^{-2} \SCP{v(x_1 - x_2) \Psi_{N,0}}{H_f \Psi_{N,0}}  
\nonumber \\
&= \left( 1 -  N^{-1} \right) \scp{\varphi_0}{\left( v * \abs{\varphi_0}^2 \right) \varphi_0} \norm{u_0}^2.
\end{align}
This shows the last inequality and altogether we obtain
\begin{align}
\abs{\SCP{ N^{-1} H_N \Psi_{N,0}}{ N^{-1} H_N \Psi_{N,0}} -\mathcal{E}^2_M[\varphi_0, u_0]} \leq  C \Lambda^4 N^{-1}  ,
\end{align}
which proves Lemma~\ref{lemma: Pauli smallness of initial beta}.
\end{proof}

\section{Proof of Theorem~\ref{theorem: Pauli main theorem}}

Let $v$ satisfy (A1), $(\varphi_t,\boldsymbol{A}(t), \boldsymbol{E}(t)) \in \mathcal{G}$ and  $\Psi_{N,t}$ be the unique solution of \eqref{eq: Pauli Schroedinger equation microscopic} with initial data $\Psi_{N,0} \in \left(L^2_s \left( \mathbb{R}^{3N} \right) \otimes \mathcal{F}_p \right) \cap \mathcal{D}(H_N)$. 
According to Lemma~\ref{lemma: Pauli dt beta} and Lemma~\ref{lemma: Pauli smallness of initial beta} there is a monotone increasing function of $\norm{\varphi_s}_{H^2(\mathbb{R}^3)}$, $\norm{\vAc (s)}_{\infty}$ and $\mathcal{E}_M[\varphi_s, u_s]$ such that
 \begin{align}
 \beta(\Psi_{N,t},\varphi_t,u_t) \leq e^{\Lambda^4 \int_0^t ds \, C(s)} \left( a_N + b_N + c_N + \Lambda/N  \right).
\end{align}
The energy $\mathcal{E}_M[\varphi_s,u_s] = \mathcal{E}_M[\varphi_0,u_0]$ is a finite constant of motion. Moreover, we have $\norm{\vAc}_{\infty} \leq \norm{\boldsymbol{A}}_{H^2(\mathbb{R}^3)}$. This displays that $C(s)$ only depends on $\norm{\varphi_{s}}_{H^2(\mathbb{R}^2)}$
and $\norm{\boldsymbol{A}}_{H^2(\mathbb{R}^3)}$. We choose for a given time $t \geq 0$ the number $N$ of charges large enough so that $\beta(\Psi_{N,t},\varphi_t,u_t) \leq 1$ and obtain
\begin{align}
\text{Tr}_{L^2(\mathbb{R}^3)} \abs{\gamma_{N,t}^{(1,0)} - \ket{\varphi_t} \bra{\varphi_t}} &\leq
  \sqrt{a_N + b_N + c_N + \Lambda/N } \, e^{\Lambda^4 \int_0^t ds \, C(s)}
\nonumber\\
\text{Tr}_{\mathfrak{h}} \abs{\gamma_{N,t}^{(0,1)} - \ket{u_t}\bra{u_t}}  &\leq  
\sqrt{a_N + b_N + c_N + \Lambda/N } \, 6 ( 1 + \norm{u_t}_{\mathfrak{h}} ) e^{\Lambda^4 \int_0^t ds \, C(s)} 
\end{align}
by Lemma~\ref{lemma: Pauli relation between beta and reduced density matrices}. Then, we recall \eqref{eq: Pauli dependence of the field energy on A and E} and derive
\begin{align}
\text{Tr}_{L^2(\mathbb{R}^3)} \abs{\gamma_{N,t}^{(1,0)} - \ket{\varphi_t} \bra{\varphi_t}} &\leq
\sqrt{a_N + b_N + c_N + N ^{-1} } \, \Lambda e^{\Lambda^4 \int_0^t ds \, C(s)} 
\nonumber\\
\text{Tr}_{\mathfrak{h}} \abs{\gamma_{N,t}^{(0,1)} - \ket{u_t}\bra{u_t}}  &\leq  
\sqrt{a_N + b_N + c_N + N^{-1} } \,  \Lambda C(s) e^{\Lambda^4 \int_0^t ds \, C(s)} 
\end{align}
where $C(s)$ depends on $\norm{\varphi_{s}}_{H^2(\mathbb{R}^2)}$, $\norm{\boldsymbol{A}}_{H^2(\mathbb{R}^3)}$ and $\norm{\boldsymbol{E}}_{L^2(\mathbb{R}^3)}$.
For initial states of product type $\Psi_{N,0} = \varphi_0^{\otimes N} \otimes W(\sqrt{N} \alpha_0) \Omega$ this becomes
\begin{align}
\text{Tr}_{L^2(\mathbb{R}^3)} \abs{\gamma_{N,t}^{(1,0)} - \ket{\varphi_t} \bra{\varphi_t}} &\leq
 N^{-1/2}  \Lambda^2 e^{\Lambda^4 \int_0^t ds \, C(s)}
\nonumber\\
\text{Tr}_{\mathfrak{h}} \abs{\gamma_{N,t}^{(0,1)} - \ket{u_t}\bra{u_t}}  &\leq
 N^{-1/2}  \Lambda^2 C(s) e^{\Lambda^4 \int_0^t ds \, C(s)} .
\end{align}

\section{Appendix}
\label{section: appendix}

\begin{lemma}
\label{lemma: Pauli field operator times p or q against functional}
Let $\Psi_{N,t} \in \left(L^2_s \left( \mathbb{R}^{3N} \right) \otimes \mathcal{F}_p \right) \cap \mathcal{D}(H_N)$ and $(\varphi_t,\boldsymbol{A}(t), \boldsymbol{E}(t)) \in \mathcal{G}$. Then

\begin{align}
\int d^3y \,  \norm{ \left( N^{-1/2} \vEp(y)  - \vEpc(y,t) \right) \Psi_{N,t}}^2
&= \SCP{\Psi_{N,t}}{ \vEdm(y,t) \vEdp(y,t) \Psi_{N,t}}_{;y} 
 \leq \beta^b(t) .
\end{align}
For $\vG \in \left\{ \vAdp, \vAdm, \vEdp, \vEdm \right\}$ one obtains
\begin{align}
\norm{\vG(y,t) \nabla_1 p_1 \Psi_{N,t}}_{;y}^2 
&\leq  C \norm{\nabla_1 \varphi_t}_{L^2(\mathbb{R}^3)}^2   \norm{\vG(y,t) \Psi_{N,t}}_{;y}^2,
\nonumber \\
\norm{\vG(x_ 1,t) p_1 \Psi_{N,t}}^2 
&\leq C \norm{\varphi_t}_{L^{\infty}(\mathbb{R}^3)}^2  \norm{\vG(y,t) \Psi_{N,t}}_{;y}^2.
\end{align}
\end{lemma}
\begin{proof}
The first inequality is proven by
\small
\begin{align}
\norm{\vEdp(y,t) \Psi_N}_{;y}^2 
& =  1/2
 \sum_{\lambda=1,2} \int d^3k \, \tilde{\kappa}(k) \abs{k}^{1/2} \vep_{\lambda}(k) 
 \sum_{\mu =1,2} \int d^3l \, \tilde{\kappa}(l)  \abs{l}^{1/2}   \vep_{\mu}(l)   \int d^3y \, e^{i(l-k)y}  
 \nonumber \\
&  \times  \SCP{\left( N^{-1/2} a(k,\lambda) - \alpha_t(k,\lambda) \right) \Psi_N}{\left(  N^{-1/2} a(l,\mu) - \alpha_t(l,\mu) \right) \Psi_N}
\nonumber \\
&=  (2 \pi)^3/2 \int d^3l \, \abs{\tilde{\kappa}(k)}^2 
\abs{k} \sum_{\lambda,\mu} \vep_{\lambda}(k) \vep_{\mu}(k)  \times
\nonumber \\
& \times  \SCP{\left( N^{-1/2} a(k,\lambda) - \alpha_t(k,\lambda) \right) \Psi_N}{\left(  N^{-1/2} a(k, \mu) - \alpha_t(k,\mu) \right) \Psi_N} 
\nonumber \\
&= 1/2 \sum_{\lambda=1,2} \int_{\abs{k} \leq \Lambda}  d^3k \, \abs{k}
\norm{\left( N^{-1/2} a(k,\lambda) - \alpha_t(k,\lambda) \right) \Psi_N}^2  \leq \beta^b  .
\end{align}
\normalsize
We continue with
\begin{align}
\norm{\vG(y,t) \nabla_1 p_1 \Psi_N}_{;y}^2  
&= \int d^3y \, \norm{\vG(y,t) \nabla_1 p_1 \Psi_N}^2
= \int d^3y \, \norm{\nabla_1 p_1 \vG(y,t)  \Psi_N}^2
\nonumber\\
&= \int d^3y \, \SCP{\vG(y,t)  \Psi_N}{p_1 (- \Delta_1) p_1 \vG(y,t)  \Psi_N}.
\end{align}
So if we use $p_1 (- \Delta_1) p_1 = p_1 \norm{\nabla \varphi}_2^2$ we get
\begin{align}
\norm{\vG(y,t) \nabla_1 p_1 \Psi_N}_{;y}^2 
&= \norm{\nabla \varphi}^2 \int d^3y \, \norm{p_1 \vG(y,t)  \Psi_N}^2 
\leq \norm{\nabla \varphi}^2 \int d^3y \, \norm{\vG(y,t)  \Psi_N}^2. 
\end{align}
In the same way, we apply
\begin{align}
p_1   \big( \vG(x_1,t) \big)^* \vG(x_1,t) p_1 &= p_1 \int d^3y \, \abs{\varphi(y)}^2 \big( \vG(y,t) \big)^* \vG(y,t)
\end{align}
to show the third inequality
\begin{align}
\norm{\vG(x_ 1,t) p_1 \Psi_N}^2
&= \int d^3y \, \abs{\varphi(y)}^2 \SCP{\Psi_N}{p_1  \big( \vG(y,t) \big)^* \vG(y,t) \Psi_N}
\nonumber \\
&= \int d^3y \, \abs{\varphi(y)}^2  \norm{p_1  \vG(y,t) \Psi_N}^2
\leq \norm{\varphi}_{\infty}^2 \norm{\vG(y,t) \Psi_N}_{;y}^2.
\end{align}

\end{proof}

\section*{Acknowledgments}

We thank Dirk Andr\'{e} Deckert, Jan Derezi\'{n}ski, Detlef D\"urr, Marco Falconi, Maximilian Jeblick, Vytautas Matulevi\v{c}ius,  and Alessandro Michelangeli for many helpful remarks. We are deeply grateful to Vytautas Matulevi\v{c}ius for valuable discussions at the early stage of this project and to Alessandro Michelangeli for helpful remarks concerning the Maxwell-Schr\"odinger system. 
N.L. gratefully acknowledges financial support by the Cusanuswerk and the European Research Council (ERC) under the European Union's Horizon 2020 research and innovation programme (grant agreement No 694227).
The article appeard in slightly different form in one of the author's (N.L.) PhD thesis \cite{leopold2}.

{}

\end{document}